\documentclass[preprint,12pt]{elsarticle}

 \usepackage{graphics}

\usepackage{amssymb}

\usepackage{algorithm,algorithmic,multirow}

\journal{a journal}

\newcommand{\fpara}[4]{ 
\begin{picture}(0,0)%
\setlength{\unitlength}{1pt}%
\put(#1,#2){\rotatebox{#3}{\raisebox{0mm}[0mm][0mm]{%
\makebox[0mm]{$\left.\rule{0mm}{#4pt}\right($}}}}%
\end{picture}}
\newcommand{\bpara}[4]{ 
\begin{picture}(0,0)%
\setlength{\unitlength}{1pt}%
\put(#1,#2){\rotatebox{#3}{\raisebox{0mm}[0mm][0mm]{%
\makebox[0mm]{$\left)\rule{0mm}{#4pt}\right.$}}}}%
\end{picture}}

\newtheorem{lemma}{Lemma}
\newtheorem{observation}{Observation}

\newtheorem{thm}{Theorem}

\newenvironment{proof}{\noindent {\em Proof.}}{\qed\bigskip}

\newcommand{\opt}{{\mathit{opt}}}
\newcommand{\angResl}{{\mathit{angResl}}}

\newcommand{\APX}{{\mathit{APX}}}
\newcommand{\LB}{{\mathit{LB}}}

\begin{document}

\begin{frontmatter}



\title{Complexity Analysis of Balloon Drawing for Rooted Trees}

\author{Chun-Cheng~Lin\fnref{TMUE}}

\author{Hsu-Chun~Yen\fnref{NTU,KNU}\corref{cor1}}
\ead{yen@cc.ee.ntu.edu.tw} \cortext[cor1]{Corresponding author.}

\author{Sheung-Hung Poon\fnref{NTHU}}

\author{Jia-Hao Fan\fnref{NTU}}

\fntext[TMUE]{Dept. of Computer Science, Taipei Municipal University of Education, Taipei, Taiwan, 
R.O.C.}%
\fntext[NTU]{Dept. of Electrical Engineering, National Taiwan University, Taipei, Taiwan, 
R.O.C.}%
\fntext[KNU]{Dept. of Computer Science, Kainan University, Taoyuan, Taiwan, 
R.O.C.}%
\fntext[NTHU]{Dept. of Computer Science, National Tsing Hua University, Hsinchu, Taiwan, 
R.O.C.}%



\begin{abstract}
In a {\em  balloon drawing} of a tree, all the children under the
same parent are placed on the circumference of the circle centered
at their parent, and the radius of the circle centered at each node
along any path from the root reflects the number of descendants
associated with the node. Among various styles of tree drawings
reported in the literature, the balloon drawing enjoys a desirable
feature of displaying tree structures in a rather balanced fashion.
For each internal node in a balloon drawing, the ray from the node
to each of its children divides the wedge accommodating the subtree
rooted at the child into two sub-wedges. Depending on whether the
two sub-wedge angles are required to be identical or not, a balloon
drawing can further be divided into two types: {\em even sub-wedge}
and {\em uneven sub-wedge} types. In the most general case, for any
internal node in the tree there are two dimensions of freedom that
affect the quality of a balloon drawing: (1) altering the order in
which the children of the node appear in the drawing, and (2) for
the subtree rooted at each child of the node, flipping the two
sub-wedges of the subtree. In this paper, we give a comprehensive
complexity analysis for optimizing balloon drawings of rooted trees
with respect to {\em angular resolution}, {\em aspect ratio} and
{\em standard deviation of angles} under various drawing cases
depending on whether the tree is of even or uneven sub-wedge type
and whether (1) and (2) above are allowed. It turns out that some
are NP-complete while others can be solved in polynomial time. We
also derive approximation algorithms for those that are intractable
in general.
\end{abstract}

\begin{keyword}
tree drawing \sep graph drawing \sep graph algorithms



\end{keyword}

\end{frontmatter}


\section{Introduction}
\label{sec:introduction}

Graph drawing addresses the issue of constructing geometric
representations of graphs in a way to  gain better understanding and
insights into the graph structures.  Surveys on graph drawing can
be found in \cite{BETT1999,KW2001}. If the given data is
hierarchical (such as a file system
), then it can often be expressed as a rooted tree. Among existing
algorithms in
the literature for drawing rooted trees, the work of 
\cite{RT1981} developed a popular method for drawing binary trees.
The idea behind \cite{RT1981} is to recursively draw the left and
right subtrees independently in a bottom-up manner, then shift the
two drawings along the $x$-direction as close to each other as
possible while centering the parent of the two subtrees one level up
between their roots. Different from the conventional `triangular'
tree drawing of \cite{RT1981}, $hv$-drawings \cite{Shiloach1976},
radial drawings \cite{Eades1992} and balloon drawings
\cite{CK1995,JP1998,KY1993,LY2007,MH1998}
are also popular for visualizing hierarchical graphs. Since the
majority of  algorithms for drawing rooted trees take linear time,
rooted tree structures are suited to be used in an environment in
which  real-time interactions with users are frequent.

Consider Figure~\ref{fgIllustration} for an example. A {\em balloon
drawing} \cite{CK1995,JP1998,LY2007} of a rooted tree is a drawing
having the following properties:
\begin{itemize}

\item all the children under the same parent are placed on the
circumference of the circle centered at their parent;

\item there exist no edge crossings in the drawing;

\item the radius of the circle centered at each node along any path
from the root node reflects the number of descendants associated
with the node (i.e., for any two edges on a path from the root node,
the farther from the root an edge is, the shorter its drawing length
becomes).
\end{itemize}
In the balloon drawing of a tree, each subtree resides in a {\em
wedge} whose end-point is the parent node of the root of the
subtree. The ray from the parent node to the root of the subtree
divides the wedge into two {\em sub-wedges}. Depending on whether
the two sub-wedge angles are required to be identical or not, a
balloon drawing can further be divided into two types: drawings with
{\em even sub-wedges} (see Figure~\ref{fgIllustration}(a)) and
drawings with {\em uneven sub-wedges} (see
Figure~\ref{fgIllustration}(b)). One can see from the transformation
from Figure~\ref{fgIllustration}(a) to
Figure~\ref{fgIllustration}(b) that a balloon drawing with uneven
sub-wedges is derived from that with even sub-wedges by shrinking
the drawing circles in a bottom-up fashion so that the drawing area
is as small as possible \cite{LY2007}. Another way to differentiate
the two is that for the even sub-wedge case, it is required that the
position of the root of a subtree coincides with the center of the
enclosing circle of the subtree.

\begin{figure}[tbp]
\centering \scalebox{0.625}{\includegraphics{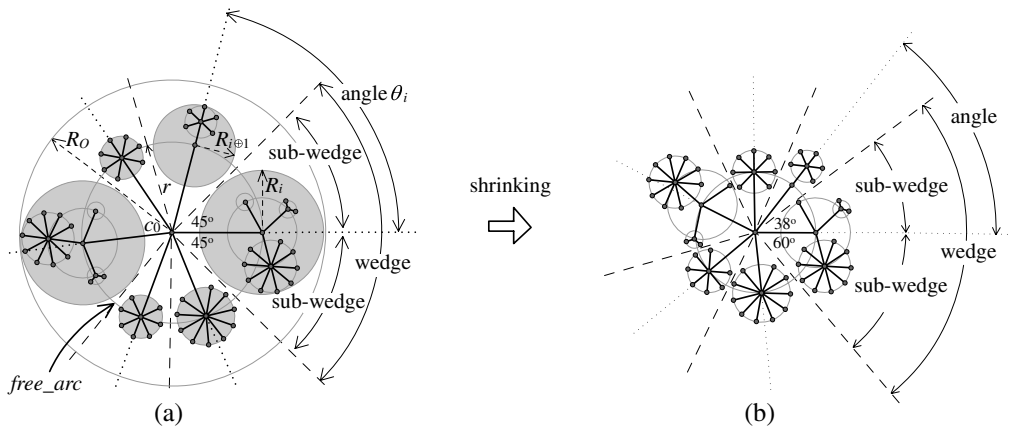}}%
\centering \caption{Illustration of balloon drawings with (a) even
sub-wedges and (b) uneven sub-wedges, where each node is drawn by a point; each edge is drawn by a solid straight line
segment; the center of the largest circle in a balloon drawing is
the root node.} \label{fgIllustration}
\end{figure}



{\em Aesthetic criteria} specify graphic structures and properties
of drawing, such as minimizing number of edge crossings or bends,
minimizing area, and so on, but the problem of simultaneously
optimizing those criteria is, in many cases, NP-hard. The main
aesthetic criteria on the angle sizes in balloon drawings are {\em angular resolution},
{\em aspect ratio}, and {\em standard deviation of angles}. Note that this paper mainly concerns the angle sizes, while it is interesting to investigate other aesthetic criteria, such as the drawing area, total edge length, etc. Given a
drawing of tree $T$, an angle formed by the two adjacent edges
incident to a common node $v$ is called an angle incident to node
$v$. Note that an angle in a balloon drawing consists of two
sub-wedges which belong to two different subtrees, respectively (see
Figure~\ref{fgIllustration}). With respect to a node $v$, the {\em
angular resolution} is the smallest angle incident to node $v$, the
{\em aspect ratio} is the  ratio of the largest angle to the
smallest angle incident to node $v$, and the {\em standard deviation
of angles} is a statistic used as a measure of the dispersion or
variation in the distribution of angles, equal to the square root of
the arithmetic mean of the squares of the deviations from the
arithmetic mean.

The {\em angular resolution} (resp., {\em aspect ratio}; {\em
standard deviation of angles}) of a drawing of $T$ is defined as the
minimum angular resolution (resp., the maximum aspect ratio; the
maximum standard deviation of angles) among all nodes in $T$. The
angular resolution (resp., aspect ratio; standard deviation of
angles) of a tree drawing is in the range of $(0^{\circ},
360^{\circ})$ (resp., $[1,\infty)$ and $[0,\infty)$). A tree layout
with a large angular resolution can easily  be identified by eyes,
while a tree layout with a small aspect ratio or standard deviation
of angles often enjoys a very balanced view of tree drawing. It is
worthy of pointing out the fundamental difference between aspect
ratio and standard deviation. The aspect ratio only concerns the
deviation between the largest and the smallest angles in the
drawing, while the standard deviation deals with the deviation of
all the angles.

With respect to a balloon drawing of a rooted tree, changing the
order in which the children of a node are listed or flipping the two
sub-wedges of a subtree affects the quality of the drawing. For
example, in comparison between the two balloon drawings of a tree
under different tree orderings respectively shown in
Figures~\ref{fgExperiments}(a) and \ref{fgExperiments}(b), we
observe that the drawing in Figure~\ref{fgExperiments}(b) displays
little variations of angles, which give a very balanced drawing.
Hence some interesting questions arise: {\em How to change the tree
ordering or flip the two sub-wedge angles of each subtree such that
the balloon drawing of the tree has the maximum angular resolution,
the minimum aspect ratio, and the minimum standard deviation of
angles?}

\begin{figure}[tbp]
\centering \scalebox{1.0}{\includegraphics{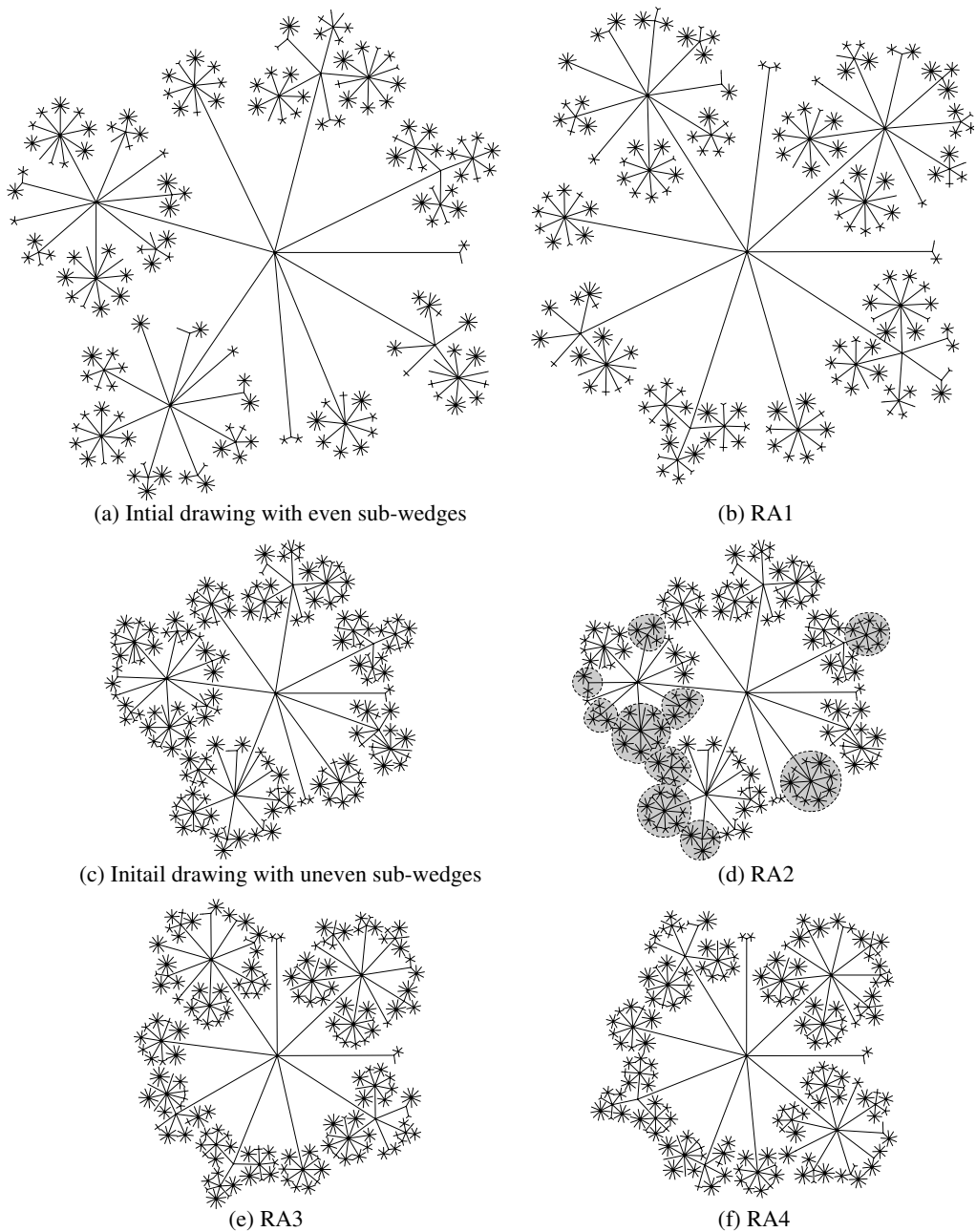}}%
\centering \caption{An experimental example, where (a) and (c) are
initial balloon drawings with even and uneven sub-wedges,
respectively; (b), (d), (e) and (f) achieve the optimality of RA1,
RA2, RA3 and RA4, respectively. Note that the differences of (d) from (c) are encompassed by shaded regions.} \label{fgExperiments}
\end{figure}

Throughout the rest of this paper, we let {\em RE}, {\em RA}, and
{\em DE} denote the problems of optimizing angular resolution,
aspect ratio, and standard deviation of angles, respectively. In
this paper, we investigate the tractability of the RE, RA, and DE
problems in a variety of cases, and our main results are listed in
Table~\ref{tb-results}, in which
trees with `flexible' (resp., `fixed')  uneven sub-wedges  refer  to
the case when sub-wedges of subtrees are (resp., are not)  allowed
to flip;
a `semi-ordered' tree is an unordered tree where
only the circular ordering of the children of each node is fixed, without specifying if this ordering is clockwise or counterclockwise in the drawing. Note that a semi-ordered tree allows to flip uneven sub-wedges in the drawing, because flipping sub-wedges of a node in the bottom-up fashion of the tree does not modify the circular ordering of its children.
See Figure~\ref{fgExperiments} for an experimental example
with the drawings which achieve the optimality of RA1--RA4. In
Table~\ref{tb-results}, with the exception of RE1 and RA1 (which
were previously obtained by Lin and Yen in \cite{LY2007}), all the
remaining results are new.
We also give
2-approximation algorithms for RA3 and RA4, and
$O(\sqrt{n})$-approximation algorithms for DE3 and DE4. Finding
improved approximation bounds for those intractable problems remains
an interesting open question.

\begin{table}[tbp]
  \caption{The time complexity for optimizing main aesthetic criteria of balloon drawing.}
  \label{tb-results}
  \begin{center}{\scriptsize  \tabcolsep=5pt
\begin{tabular}{clcccc}
  \hline
  \multicolumn{2}{c}{case}            & \multicolumn{1}{c}{aesthetic criterion} & denotation  & complexity      & reference \\
  \hline
  \multirow{2}{0.5cm}{C1:} & \fpara{0}{-8}{0}{14} \ unordered trees with \ \ \ \ \ \ \ \ \bpara{0}{-8}{0}{14}         & angular resolution  & RE1         &  $O(n \log n)$   & \cite{LY2007} \\%
                           & \ \ \ \ \ \ \ \ \ \ \ \ \ \ \ even sub-wedges             & aspect ratio        & RA1         & $O(n \log n)$   & \cite{LY2007} \\%
                           &                                 & standard deviation  & DE1         & $O(n \log n)$$^*$ & [Thm \ref{thm-DE1}]  \\
  \hline
  \multirow{2}{0.5cm}{C2:} & \fpara{0}{-8}{0}{14} \ semi-ordered  trees with \ \ \ \ \bpara{0}{-8}{0}{14}    & angular resolution  & RE2         &  $O(n)$$^*$        & [Thm \ref{thm-RE2}] \\%
                           & \ \ flexible uneven sub-wedges  & aspect ratio        & RA2         & $O(n^2)$$^*$ & [Thm \ref{thm-RA2}] \\%
                           &                                 & standard deviation  & DE2         & $O(n)$$^*$        & [Thm \ref{thm-DE2}]  \\
  \hline
  \multirow{2}{0.5cm}{C3:} & \fpara{0}{-8}{0}{14} \ unordered trees with \ \ \ \ \ \ \ \ \bpara{0}{-8}{0}{14}           & angular resolution  & RE3         & $O(n \log n)$$^*$ & [Thm \ref{thm-RE3}] \\%
                           & \ \ fixed \ \ \ uneven sub-wedges     & aspect ratio        & RA3         & NPC$^*$         & [Thm \ref{thm-RA3}, \ref{thm-RA3-approx}] \\%
                           &                                 & standard deviation  & DE3         & NPC$^*$         & [Thm \ref{thm-DE3}, \ref{thm-DE3-approx2}] \\%
  \hline
  \multirow{2}{0.5cm}{C4:} & \fpara{0}{-8}{0}{14} \ unordered trees with \ \ \ \ \ \ \ \ \bpara{0}{-8}{0}{14}           & angular resolution  & RE4         & $O(n \log n)$$^*$ & [Thm \ref{thm-RE4}] \\%
                           & \ \ flexible uneven sub-wedges  & aspect ratio        & RA4         & NPC$^*$         & [Thm \ref{thm-RA4}, \ref{thm-RA4-approx}] \\%
                           &                                 & standard deviation  & DE4         & NPC$^*$         & [Thm \ref{thm-DE4}, \ref{thm-DE4-approx2}] \\%
  \hline
\end{tabular}
\par\smallskip\parbox{.9\textwidth}{$^*$The marked entries are
the contributions of this paper. Note that earlier results reported
in \cite{LY2007} for  RE2 and RA2 require $O(n^{2.5})$ time.} }
\end{center}
\end{table}



The rest of the paper is organized as follows. Some preliminaries
are given in Section~\ref{sec:preliminary}. The problems for cases
C1 and C2 are investigated in Section~\ref{sec:C1C2}. The problems
for cases C3 and C4 are investigated in Section~\ref{sec:C3C4}. The
approximation algorithms for those intractable problems are given in
Section~\ref{sec:approx}. Finally, a conclusion is given in
Section~\ref{sec:conclusion}.

\section{Preliminaries}
\label{sec:preliminary}

In this section, we first introduce two conventional models of
balloon drawing, then define our concerned problems, and finally
introduce some related problems.




\subsection{Two Models of Balloon Drawing}

There exist two models in the literature for generating {\em balloon
drawings} of trees. 
Given a node $v$, let $r(v)$
be the radius of the drawing circle centered at $v$. 
If we require that $r(v)$ = $r(w)$ for arbitrary two nodes $v$ and
$w$ that are of the same depth from the root of the tree, then such
a drawing is called a balloon drawing under the {\em fractal model}
\cite{KY1993}. 
The fractal drawing of a tree structure means that if $r_m$ and
$r_{m-1}$ are the lengths of edges at depths $m$ and $m-1$,
respectively,  then
$r_{m} = \gamma \times r_{m-1}$ 
where $\gamma$ is the predefined ratio ($0 < \gamma < 1$) associated
with the drawing under the fractal model.  Clearly, edges at the
same depth have the same length in a fractal drawing.

Unlike the fractal model, the {\em subtrees with nonuniform sizes}
(abbreviated as {\em SNS}) model \cite{CK1995,JP1998} allows
subtrees associated with  the same parent to reside in circles of
different sizes (see also Figure~\ref{fgIllustration}(a)), and hence
the drawing based on this model often results in a clearer display
on large subtrees than that under the fractal model.
Given a rooted ordered tree $T$ with $n$ nodes, a balloon drawing
under the SNS model can be obtained in $O(n)$ time (see
\cite{CK1995,JP1998}) in a bottom-up fashion by computing the edge
length $r$ and the angle $\theta_{i}$ between two adjacent
edges respectively according to 
$r = C/(2\pi) \cong
(2\sum\nolimits_{i}R_{i})/(2\pi)$ 
and $\theta_{i} \cong (R_{i} + free\_arc +
R_{i+1})/r$ 
(see Figure~\ref{fgIllustration}(a))
where $r$ 
is the radius of the inner circle centered at node $c_0$; $C$ is the
circumference of the inner circle; $R_{i}$ is the radius of the
outer circle enclosing all subtrees of the $i$-th child of $c_0$,
and $R_{O}$ is the radius of the outer circle enclosing all subtrees
of $c_0$; since there exists a gap between $C$ and the sum
of all diameters
, we can distribute to every $\theta_{i}$ the gap between them
evenly, which is called a free arc, denoted by $free\_arc$.

Note that the balloon drawing under the SNS model is our so-called
balloon drawing with even sub-wedges. A careful examination reveals
that the area of a balloon drawing with even sub-wedges (generated
by the SNS model) may be reduced by shrinking the free arc between
each pair of subtrees and shortening the radius of each inner circle
in a bottom-up fashion \cite{LY2007}, by which we can obtain a
smaller-area balloon drawing with uneven sub-wedges (e.g., see the
transformation from Figure~\ref{fgIllustration}(a) to
Figure~\ref{fgIllustration}(c)).

\subsection{Notation and Problem Definition}

In what follows, we introduce some notation, used in the rest of
this paper. A {\em circular permutation} $\pi$ is expressed as: $\pi
= \langle\pi_1, \pi_2, ..., \pi_n\rangle$ where for $i = 1, 2, ...,
n$, $\pi_i$ is placed along a circle in a counterclockwise
direction. Note that $\pi_n$ is adjacent to $\pi_1$; $i \oplus 1$
denotes $i + 1 \ (mod \ n)$; $i\ominus 1$ denotes $i - 1 \ (mod \
n)$. Due to the hierarchical nature of trees  and the ways the
aesthetic criteria
(measures) for balloon drawings are defined, an algorithm optimizing a {\em star graph} 
can be applied repeatedly to a general tree in a bottom-up fashion
\cite{LY2007}, yielding an optimum solution with respect to a given
aesthetic criterion. Thus, it suffices to consider the balloon
drawing of a star graph when we discuss these problems.

A star graph is characterized by a root node $c_0$ together with its
$n$ children $c_{1}, ..., c_{n}$, each of which is the root of a
subtree located entirely in a {\em wedge}, as shown in
Figure~\ref{fgIllustration}(a) (for the even sub-wedge type) and
Figure~\ref{fgUnevenNotation} (for the uneven sub-wedge type).
In what follows, we can only see Figure~\ref{fgUnevenNotation}
because the even sub-wedge type can be viewed as a special case of
the uneven sub-wedge type. The ray from $c_{0}$ to $c_{i}$ further
divides the associated wedge into two sub-wedges $SW_{i, 0}$ and
$SW_{i, 1}$ with sizes of
angles $w_{0}(i)$ and $w_{1}(i)$, respectively. 
Note that $w_{0}(i)$ and $w_{1}(i)$ need not be equal in general. An
{\em ordering} of $c_{0}$'s children is simply a circular
permutation $\sigma = \langle \sigma_1, \sigma_2, ..., \sigma_n
\rangle$, in which $\sigma_i \in \{1, 2, ..., n\}$ for each $i$.

\begin{figure}[tbp]
\centering \scalebox{0.625}{\includegraphics{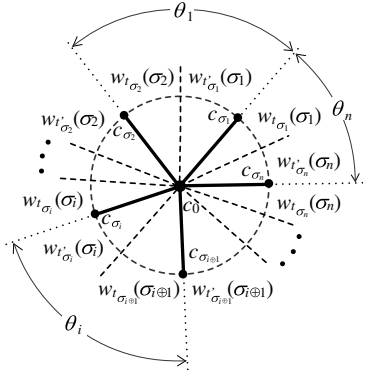}}%
\centering \caption{ Notations used in a balloon drawing of a star
graph with uneven sub-wedges.} \label{fgUnevenNotation}
\end{figure}

There are two dimensions of freedom affecting the quality of a
balloon drawing for a star graph. The first is concerned with the
ordering in which the children of the root node $c_0$ are drawn.
With a given ordering, it is also possible to alter the order of
occurrences of the two sub-wedges associated with each child of the
root. With respect to child $c_{i}$ and its two sub-wedges $SW_{i,
0}$ and $SW_{i, 1}$, we use $t_{i} \in \{0, 1\}$ to denote the index
of the first sub-wedge encountered  in a counterclockwise traversal
of the drawing. For convenience, we let $t'_{i} = 1 - t_{i}$. We
also write $t=(t_1,...,t_n)$ ($t_{i} \in \{0, 1\}, 1 \leq i \leq
n$), which is called the {\em sub-wedge assignment} (or simply {\em
assignment}). As shown in Figure~\ref{fgUnevenNotation}, the
sequence of sub-wedges encountered along the cycle centered at $c_0$
in a counterclockwise direction can be expressed as:
\begin{eqnarray}
\label{E-subWedge-uneven}
\begin{array}{ccccc}
\langle \ \underbrace{w_{t_{\sigma_1}}(\sigma_1),
w_{t'_{\sigma_1}}(\sigma_1)}, &
..., &
\underbrace{w_{t_{\sigma_i}}(\sigma_i), w_{t'_{\sigma_i}}(\sigma_i)}, & ..., & \underbrace{w_{t_{\sigma_n}}(\sigma_n), w_{t'_{\sigma_n}}(\sigma_n)}\ \rangle.\\
c_{\sigma_1} & 
... & c_{\sigma_i} & ... & c_{\sigma_n}
\end{array}
\end{eqnarray}

If $w_0(i) = w_1(i)$ for each $i\in\{1,...,n\}$, then the drawing is
said to be of {\em even sub-wedge type}; otherwise, it is of {\em
uneven sub-wedge type}. 
As mentioned earlier, the order of the two sub-wedges associated
with a child (along the counterclockwise direction) affects the
quality of a drawing in the uneven sub-wedge case. For the case of
uneven sub-wedge type, if the assignment $t$ is given {\em a
priori},
then the drawing is said to be of {\em fixed} uneven sub-wedge type;
otherwise, of {\em flexible} uneven sub-wedge type (i.e., $t$ is a
design parameter).

As shown in Figure~\ref{fgUnevenNotation}, with respect to an
ordering $\sigma$ and an assignment $t$ in circular
permutation~(\ref{E-subWedge-uneven}), $c_{\sigma_{i}}$ and
$c_{\sigma_{i \oplus 1}}$, $1 \leq i \leq n$, are neighboring nodes,
and the size of the angle formed by the two adjacent edges
$\overrightarrow{c_{0}c_{\sigma_i}}$ and
$\overrightarrow{c_{0}c_{\sigma_{i \oplus 1}}}$ is
$\theta_i=w_{t_i'}(\sigma_i)+w_{t_{i\oplus 1}}(\sigma_{i\oplus 1})$.
Hence, the {\em angular resolution} (denoted by
$AngResl_{\sigma,t}$), the {\em aspect ratio} (denoted by
$AspRatio_{\sigma,t}$), and the {\em standard deviation of angles}
(denoted by $StdDev_{\sigma,t}$) can be formulated as {\footnotesize
\begin{eqnarray}
&&AngResl_{\sigma,t} = \min_{1 \leq i \leq n} \theta_i
 = \min_{1\leq i\leq n} \{ w_{t_i'}(\sigma_i)+w_{t_{i\oplus 1}}(\sigma_{i\oplus1})
 \}
 ; \nonumber\\
&&AspRatio_{\sigma,t} =   \frac{\max_{1 \leq i \leq n}
\theta_i}{\min_{1 \leq i \leq n} \theta_i} = \frac{\max_{1\leq i\leq
n} \{ w_{t_i'}(\sigma_i)+w_{t_{i\oplus 1}}(\sigma_{i\oplus1})
 \}}{\min_{1\leq i\leq n} \{ w_{t_i'}(\sigma_i)+w_{t_{i\oplus 1}}(\sigma_{i\oplus1})
 \}};\nonumber\\
&&StdDev_{\sigma,t} =
\sqrt{\frac{\sum_{i=1}^n\theta_i^2}{n}-\left(\frac{\sum_{i=1}^n\theta_i}{n}\right)^2}
\nonumber\\
\label{E-stdDev2} && =
\sqrt{\frac{\sum_{i=1}^n(w_{t_i'}(\sigma_i)^2+w_{t_{i\oplus
1}}(\sigma_{i\oplus 1})^2)}{n}+\frac{2\sum_{i=1}^n
w_{t_i'}(\sigma_i)w_{t_{i\oplus 1}}(\sigma_{i\oplus
1})}{n}-\left(\frac{2\pi}{n}\right)^2}.
\end{eqnarray}
}We observe that the first and third terms inside the square root of
the above equation are constants for any circular permutation
$\sigma$ and assignment $t$, and hence,  the second term inside the
square root is the dominant factor as far as $StdDev_{\sigma,t}$ is
concerned. We denote by $SOP_{\sigma,t}$ the sum of products of
sub-wedges,
 which can be expressed as:
\[
SOP_{\sigma,t} = \sum_{i=1}^n w_{t_i'}(\sigma_i)w_{t_{i\oplus 1}}(\sigma_{i\oplus 1}).
\]

We are now in a position to define the RE, RA and DE problems in
Table~\ref{tb-results} for four cases (C1, C2, C3, and C4) in a
precise manner. The four cases depend on whether the circular
permutation $\sigma$ and the assignment $t$  in a balloon drawing
are fixed (i.e., given a priori) or flexible (i.e., design
parameters). For example, case C3 allows an  arbitrary ordering of
the children (i.e., the tree is unordered), but  the relative
positions of the two sub-wedges associated with a child node are
fixed (i.e., flipping is not allowed). The remaining three cases are
easy to understand.

We consider the most flexible case, namely,  C4, for which both
$\sigma$ and $t$ are design parameters, which can be chosen from
 the set $\Sigma$ of all circular permutations of $\{1, ..., n\}$ and
  the set $\mathbb{T}$ of all $n$-bit binary strings, respectively.
The RE and RA problems, respectively,  are concerned with finding
$\sigma$ and $t$ to achieve the following:
\begin{eqnarray}
optAngResl = \max_{\sigma \in \Sigma; t \in \mathbb{T}} \{AngResl_{\sigma,t}\}; 
optAspRatio = \min_{\sigma \in \Sigma; t \in \mathbb{T}} \{AspRatio_{\sigma,t}\}\texttt{.}\nonumber%
\end{eqnarray}
The DE problem is concerned with finding $\sigma$ and $t$ to achieve
the following:
\begin{eqnarray}
optStdDev &=& \min_{\sigma \in \Sigma; t \in \mathbb{T}}
\{StdDev_{\sigma,t}\}\mbox{}.\nonumber
\end{eqnarray}
As stated earlier, $optStdDev$ is closely related to the SOP
problem, which is concerned with finding $\sigma$ and $t$ to achieve
the following:
\begin{eqnarray}
optSOP &=& \min_{\sigma \in \Sigma; t \in
\mathbb{T}}\{SOP_{\sigma,t}\}.\nonumber
\end{eqnarray}

%

\subsection{Related Problems}

Before deriving our main results, we first recall two problems,
namely,  the {\em two-station assembly line problem} (2SAL) and the
{\em cyclic two-station workforce leveling problem }(2SLW) that are
closely related to our problems of optimizing balloon drawing under
a variety of aesthetic criteria. Consider a serial assembly line
with two stations, say $ST_1$ and $ST_2$, and a set $\mathbb{J} =
\{J_1, J_2, ..., J_n\}$ of $n$ jobs. Each job $J_i = (W_{i1},
W_{i2})$ consists of two tasks processed by the two stations,
respectively, where $W_{i1}$ (resp., $W_{i2}$) is the workforce
requirement at $ST_1$ (resp., $ST_2$). Assume the processing time of
each job at each station is the same, say $\tau/n$. Consider a
circular permutation $\langle J_{\delta_1}, J_{\delta_2}, ...,
J_{\delta_n}\rangle$ of $\mathbb{J}$ where $\delta = \langle
\delta_1, \delta_2, ..., \delta_n\rangle$ is a circular permutation
of $\{1, 2, ..., n\}$. At any time point,  a single station can only
process one job. We also assume that the two stations are always
busy. During the first time range $[0,\tau/n]$, $J_{\delta_1}$ and
$J_{\delta_2}$ are processed by $ST_2$ and $ST_1$, respectively, and
the workforce requirement is $W_{\delta_12}+W_{\delta_21}$.
Similarly, for each $i$, during the time range
$[(i-1)\tau/n,i\tau/n]$, $J_{\delta_i}$ and $J_{\delta_{i\oplus 1}}$
are processed at $ST_2$ and $ST_1$ stations respectively, and the
workforce requirement is $W_{\delta_i 2}+W_{\delta_{i\oplus 1}1}$.

For example, consider $\mathbb{J} = \{J_1, J_2, J_3, J_4\}$ where
$J_1=(2,3)$, $J_2=(1,7)$, $J_3=(6,2)$, and $J_4=(4,2)$. For a
certain circular permutation $\langle J_3, J_2, J_4, J_1 \rangle$ of
$\mathbb{J}$, the workforce requirements for each period of time as
well as the jobs served at the two stations are given in
Figure~\ref{fgWorkforcePlanning}, where the largest workforce
requirement is 11; the range of the workforce requirements among all
the time periods is [3,11].

\begin{figure}[tbp]
\begin{center}
\begin{minipage}{130pt}
\scalebox{0.625}{
\includegraphics{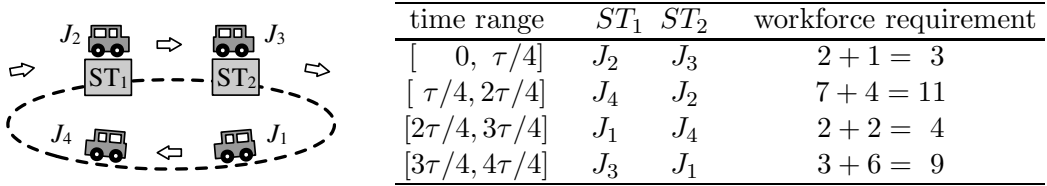}}
\end{minipage}~%
\begin{minipage}{280pt}
\begin{center}{\small  \tabcolsep=3pt
\begin{tabular}{cccc}
  \hline
time range    & \ \ \ $ST_1$ & $ST_2$ &  \ \ \  workforce requirement \\
  \hline
 $[~~~~0,~\tau/4]$     & $J_{2}$ & $J_{3}$ &    $2 + 1 = ~3$ \\
  $[~\tau/4,2\tau/4]$  & $J_{4}$ & $J_{2}$ &   $7 + 4 = 11$ \\
$[2\tau/4,3\tau/4]$ & $J_{1}$ & $J_{4}$ &    $2 + 2 = ~4$\\
  $[3\tau/4,4\tau/4]$ & $J_{3}$ & $J_{1}$ &  $3 + 6 = ~9$\\
  \hline
\end{tabular}}
\end{center}
\end{minipage}
\end{center}
\caption{An example for 2SAL and 2SLW.}\label{fgWorkforcePlanning}
\end{figure}

The $2SAL$ and $2SLW$ problems are defined as follows:
\begin{itemize}
\item {\bf 2SAL}:  Given a set  of $n$ jobs,  find a circular
permutation of the $n$ jobs such that the largest workforce
requirement is minimized.

\item {\bf 2SLW} (decision version): Given a set  of $n$ jobs and
a range $[LB, UB]$ of workforce requirements, decide whether a
circular permutation  exists  such that the workforce requirement
for each time period is between $LB$ and $UB$.
\end{itemize}
It is known that 2SAL is  solvable  in $O(n\log n)$ time
\cite{LV1997}, while 2SLW is NP-complete \cite{V2003}.

\section{Cases C1 (Unordered Trees with Even Sub-Wedges) and C2 (Semi-Ordered Trees with Flexible Uneven Sub-Wedges)}
\label{sec:C1C2}

First of all, we investigate the DE1 problem (SOP1 problem), i.e.,
finding a balloon drawing optimizing $optSOP$ for case C1 (i.e.,
unordered trees with even sub-wedges). In this case, the two
sub-wedges associated with a child node in a star graph are of the
same size. For notational convenience, we order the set of wedge
angles $\{w_0(i)+w_1(i): i = 1, \cdots, n\}$ (note that in this case
$w_0(i) = w_1(i)$ for each $i$) in ascending order as either
\begin{eqnarray}
&m_{1}, m_{2}, \cdots, m_{k-1}, m_{k}, M_{k}, M_{k-1}, \cdots,
M_{2}, M_{1} & \mbox{  if $n=2k$, or}
\label{E-subwedge-C1-even}\\
&m_{1}, m_{2}, \cdots, m_{k-1}, m_{k}, mid, M_{k}, M_{k-1}, \cdots,
M_{2}, M_{1}  &\mbox{  if $n=2k+1$,}\label{E-subwedge-C1-odd}
\end{eqnarray} for some $k$, where
$m_{i}$ (resp., $M_{i}$) is the $i$-th minimum (resp., maximum)
among all, and $mid$ is the median if the number of elements is odd.
Note that the size of each angle between two edges in the drawing
may be one of the forms $(m_a + m_b)/2$, $(m_a + M_b)/2$, $(M_a +
m_b)/2$, or $(M_a + M_b)/2$ for some $a,b \in \{1, \cdots, n\}$, and
hence, there may exist more than one angle with the same value. In
what follows, we are able to solve the DE1 problem by applying
Procedure~\ref{alg:C1}.

\begin{algorithm}{\small
\floatname{algorithm}{Procedure} \caption{$\mbox{\sc
OptBalloonDrawing-DE1}$ } \label{alg:C1}
\textbf{Input:} a star graph $S$ with $n$ child nodes of nonuniform sizes\\
\textbf{Output:} a balloon drawing  of $S$ optimizing standard deviation of angles\\
\begin{algorithmic}[1]
\STATE sort $\{w_0(i)+w_1(i): i = 1, \cdots, n\}$ in ascending order as either Equation~(\ref{E-subwedge-C1-even}), if $n = 2k$, or Equation~(\ref{E-subwedge-C1-odd}), if $n = 2k+1$
\STATE for convenience, let the child node with wedge $m_i$, $mid$ or $M_i$ be also denoted by $m_i$, $mid$ or $M_i$, respectively
\IF {$n = 2k$}
    \IF {$k$ is odd}
        \STATE output $\langle M_{1}, m_{2}, M_{3}, m_{4}, \cdots, M_{k-1}, m_k, M_k, m_{k-1}, \cdots, M_{4}, m_{3}, M_{2}, m_{1} \rangle$
    \ELSE
        \STATE output $\langle M_{1}, m_{2}, M_{3}, m_{4}, \cdots, m_{k-1}, M_k, m_k, M_{k-1}, \cdots, M_{4}, m_{3}, M_{2}, m_{1} \rangle$
    \ENDIF
\ELSE
    \IF {$k$ is odd}
        \STATE output $\langle M_{1}, m_{2}, M_{3}, m_{4}, \cdots, M_{k-1}, m_k, mid, M_k, m_{k-1}, \cdots, M_{4}, m_{3}, M_{2}, m_{1} \rangle$
    \ELSE
        \STATE output $\langle M_{1}, m_{2}, M_{3}, m_{4}, \cdots, m_{k-1}, M_k, mid, m_k, M_{k-1}, \cdots, M_{4}, m_{3}, M_{2}, m_{1} \rangle$
    \ENDIF
\ENDIF
\end{algorithmic}
}\end{algorithm}

%

\begin{thm}
\label{thm-DE1} The DE1 problem is solvable in $O(n \log n)$ time.
\end{thm}
\begin{proof}
In what follows, we show that Procedure~\ref{alg:C1}, which clearly
runs in $O(n \log n)$ time, can be applied to correctly producing
the optimum solution. We only consider an output case in
Procedure~\ref{alg:C1}:
\[ \sigma = \langle M_{1}, m_{2}, M_{3}, m_{4}, ..., M_{k-1}, m_{k}, mid,
M_{k}, m_{k-1}, ..., M_{4}, m_{3}, M_{2}, m_{1} \rangle \]
i.e.,   $n=2k+1$ and $k$ is odd
; the remaining cases are similar (in fact, simpler). Note that
$SOP_{\sigma,t} = (\sum_{i=1}^{k-1} M_i m_{i+1} + m_{k} \times mid +
mid \times M_k + \sum_{i=1}^{k-1} m_i M_{i+1} + m_1 M_1)/4$, for
this output case.

We proceed by induction on an integer number $i$, for $i=1$ to $k$,
to prove that, with respect to the SOP measure, no circular
permutations perform better than a certain circular permutation
$\delta$ which contains the sequence
\[
S_i = \left\{
\begin{array}{ll}
m_1 M_1, & \mbox{if $i = 1$;} \\%
M_{i} S_{i-1} m_{i}, & \mbox{if $i$ is even;} \\%
m_{i} S_{i-1} M_{i}, & \mbox{if $i$ is odd.}%
\end{array}
\right. 
\] If the above holds, then no circular permutations perform better than a certain circular permutation $\delta$ which
contains sequence $S_k$. That is, no circular permutations perform
better than circular permutation $\delta = \langle S_k, mid \rangle
= \sigma$, as required.

For $i=1$, we show that no circular permutations perform better than
a certain circular permutation $\delta$ which contains sequence $S_1
= m_1 M_1$. Contrarily suppose that there exists a circular
permutation $\delta'$ in which $m_1$ is not adjacent to $M_1$ so
that $SOP_{\delta'} < SOP_\delta$. We assume that $m_1$ (resp.,
$M_1$) is adjacent to $x = m_1 + l_1$ (resp., $y = m_1 + l_2$) in
$\delta'$ where $m_1 \leq x, y \leq M_1$, $x \neq y$, and $l_1, l_2
\geq 0$. W.l.o.g., let $\delta'$ be $\langle x m_1 S' y M_1 S''
\rangle$ where $S' \cup S'' = \{m_2, \cdots, m_n, mid,$ $M_n,
\cdots, M_2\} \setminus \{x, y\}$. Consider circular permutation
$\delta = \langle x y S'^R m_1 M_1 S'' \rangle$ where $S'^R$ is the
reverse of $S'$. Then $SOP_{\delta'}-SOP_{\delta}= (x m_1 + y M_1 -
xy - m_1M_1)/4 = l_2 ( M_1 - m_1 - l_1 )/4 = l_2 ( M_1 - x )/4 \geq
0$, which is a contradiction.

Suppose that no circular permutations perform better than a certain
circular permutation which contains sequence $S_{i-1}$. We show that
no circular permutation perform better than a certain circular
permutation $\delta_i$ which contains sequence $S_i$. In the
following, we only consider the case when $i$ is even (i.e., $S_i =
M_i S_{i-1} m_i$); the other case is similar.

Contrarily suppose that there exists a circular permutation
$\delta_i'$ which perform better than $\delta_i$, i.e.,
$SOP_{\delta_i'} < SOP_{\delta_i}$. By the inductive hypothesis,
$SOP_{\delta_i'} \geq SOP_{\delta_{i-1}}$ for some circular
permutation $\delta_{i-1}$ which contains sequence $S_{i-1}$.
W.l.o.g., suppose that $\delta_{i-1} =\langle S_{i-1} x_1 S' m_i x_2
S'' x_3 M_i S''' x_4 \rangle$ where $m_i \leq x_1, \cdots, x_4 \leq
M_i$ and $S' \cup S'' \cup S''' = \{m_{i+1}, \cdots, m_n, mid,$
$M_n, \cdots, M_{i+1}\} \setminus \{ x_1, \cdots, x_4\}$; the other
cases are similar. Assume $x_1 = m_i + l_1$, $\cdots$, $x_4 = m_i +
l_4$ where $l_1, \cdots, l_4 \geq 0$. Let $M_i = m_i + l_5$ where
$l_5 \geq l_j$ for each $j \in \{1, \cdots, 4\}$. Consider $\delta_i
=\langle S_{i-1} m_i S'^R x_1 x_2 S'' x_3 x_4 S'''^R M_i \rangle$.
Then $SOP_{\delta_{i-1}}-SOP_{\delta_{i}} = (M_{i-1} x_1 + m_i x_2 +
x_3 M_i + x_4 m_{i-1} - M_{i-1} m_i - x_1 x_2 - x_3 x_4 - M_i
m_{i-1})/4 = l_1 ( M_{i-1} - m_i - l_2)/4 + ( m_{i-1} - m_i - l_3 )
(l_4 - l_5)/4 = l_1 ( M_{i-1} - x_2 )/4 + ( m_{i-1} - x_3 ) (l_4 -
l_5)/4 \geq 0$. Hence, $SOP_{\delta_i'} \geq SOP_{\delta_{i-1}} \geq
SOP_{\delta_{i}}$, which is a contradiction.
\end{proof}

Now consider case C2 (semi-ordered trees with flexible uneven angles). In
this case, the ordering of children of the root, $\sigma =
(1,2,\cdots,n)$, is fixed, and only the assignment of $t=(t_1,
\cdots, t_n)$ needs to be specified. Our solutions for RE2, RA2 and
DE2 are based on dynamic programming approaches. Those results are
given as follows:

\begin{thm}
\label{thm-RE2}%
The RE2 problem can be solved in $O(n)$ time.
\end{thm}
\begin{proof}
W.l.o.g., assume $\sigma = (1,2,...,n)$. Recall from
Equation~(\ref{E-subWedge-uneven}) that if $t=(t_{1}, ..., t_{n})$
is the assignment of sub-wedges, then the sequence of  sub-wedges
encountered in a counterclockwise direction is $\langle w_{t_1}(1),
w_{t_1'}(1), w_{t_2}(2), w_{t'_2}(2), \cdots,$ $w_{t_n}(n),
w_{t'_n}(n) \rangle$. We define $f_i (w_{t_1}(1), w_{t'_i}(i))$ as
follows:
\[
 \max_{t_{j} \in \{0,1\}, 2 \leq j \leq i-1} \{\min\{
(w_{t'_1}(1)+w_{t_2}(2)), (w_{t'_2}(2)+w_{t_3}(3)), ...,
(w_{t'_{i-1}}(i-1)+w_{t_i}(i)) \}\}.
\]
That is, the solution maximizes the minimum sum of adjacent
sub-wedge pairs for the first $i$ children, given $w_{t_1}(1)$ and
$w_{t'_i}(i)$ as the outer sub-wedges of first child and $i$-th
child, respectively. Notice that $w_{t'_{i}}(i)+w_{1}(1)$ is not
included in calculating $f_i (w_{t_1}(1), w_{t'_i}(i))$, meaning
that the first child is not considered to be adjacent to the $i$-th
child. We can observe that $f_i (w_{t_1}(1), w_{t'_{i}}(i))$ can be
formulated as the following dynamic programming formula:
\[ f_i (w_{t_1}(1), w_{t'_{i}}(i)) =
\max_{t_{i-1} \in \{0,1\}} \{ \min \{ f_{i-1} (w_{t_{1}}(1),
w_{t'_{i-1}}(i-1)), \ \ w_{t'_{i-1}}(i-1) + w_{t_i}(i) \} \}.
\]
Finally,  we have:
\[
 optAngResl = \max_{t_{1}, t'_{n} \in \{0,1\}}  \{ \min \{ f_n (w_{t_{1}}(1),
 w_{t'_{n}}(n)), \ \ w_{t_{1}}(1) + w_{t'_{n}}(n) \}
     \}.
\]
It is easy to see that the above algorithm gives the correct answer
and runs in linear time.
\end{proof}


\begin{thm}
\label{thm-RA2}%
The RA2 problem can be solved in $O(n^2)$ time.
\end{thm}
\begin{proof}
Since only flipping subwedges is allowed in this case, $w_0(i)$ and $w_1(i)$ can be the neighbors of $w_0(i\oplus 1)$ and $w_1(i\oplus 1)$ for each $i \in \{ 1, \cdots, n \}$, resulting in four possible angles, i.e., $w_0(i) + w_0(i\oplus 1)$, $w_0(i) + w_1(i\oplus 1)$, $w_1(i) + w_0(i\oplus 1)$, $w_1(i) + w_1(i\oplus 1)$.
That is, $w_0(1)$ and $w_1(1)$ can be neighbored with $w_0(2)$ and $w_1(2)$; $w_0(2)$ and $w_1(2)$ can be neighbored with $w_0(3)$ and $w_1(3)$; $\cdots$ ; $w_0(n)$ and $w_1(n)$ can be neighbored with $w_0(1)$ and $w_1(1)$.
Hence, there are $O(4n)$ possible angles in total for a given sequence of
sub-wedges. 
We assume the angle $x+y$ formed by each pair $(x,y)$ of sub-wedges
to be the `largest' angle in a drawing. Then by using the dynamic
programming approach of Theorem~\ref{thm-RE2} in $O(n)$ time, we can
obtain the smallest angle $f_n(x,y)$ in the drawing, and hence the
aspect ratio for this drawing is $(x+y)/f_n(x,y)$.
Then $optApsRatio$ can be obtained after considering 
all the $O(4n)$ possible angles, so the time complexity is $O ( 4n
\times n) = O(n^2)$.
\end{proof}

Note that the use of dynamic programming allows us to reduce the
running time of RE2 and RA2 from $O(n^{2.5})$ in \cite{LY2007} to
$O(n)$ and $O(n^2)$, respectively.

\begin{thm}
\label{thm-DE2}%
The DE2 problem can be solved in $O(n)$ time.
\end{thm}
\begin{proof}
Similar to the proof in Theorem~\ref{thm-RE2}, we define
\begin{eqnarray}
g_i (w_{t_1}(1), w_{t'_i}(i)) = \min_{t_{j} \in \{0,1\}, 2 \leq j
\leq i-1} \{ w_{t'_1}(1)\times w_{t_2}(2) + w_{t'_2}(2)\times
w_{t_3}(3) + \nonumber\\
\cdots + w_{t'_{i-1}}(i-1)\times w_{t_i}(i) \}, \nonumber
\end{eqnarray}
which can be formulated as the following dynamic programming
formula:
\[ g_i (w_{t_1}(1), w_{t'_{i}}(i)) =
\min_{t_{i-1} \in \{0,1\}} \{ g_{i-1} (w_{t_{1}}(1) ,
w_{t'_{i-1}}(i-1)) \ + \ w_{t'_{i-1}}(i-1) \times w_{t_i}(i) \}.
\]
Then,  we have
\[
 optSOP = \min_{t_{1}, t'_{n} \in \{0,1\}}  \{  g_n (w_{t_{1}}(1),
 w_{t'_{n}}(n)) \ + \ w_{t_{1}}(1) \times w_{t'_{n}}(n)
     \}.
\]
Finally, by Equation~(\ref{E-stdDev2}), the solution of the DE2 problem can be obtained as follows:
\[
optStdDev = \sqrt{\frac{\sum_{i=1}^n(w_{t_i'}(\sigma_i)^2+w_{t_{i\oplus
1}}(\sigma_{i\oplus 1})^2)}{n}+optSOP-\left(\frac{2\pi}{n}\right)^2}.
\]
Note that the first and third terms inside the square root of the above equation are constants.
\end{proof}

\section{Cases C3 and C4 (Unordered Trees with Fixed/Flexible Uneven Sub-Wedges)}
\label{sec:C3C4}

In this section, we consider cases C3 and C4 (unordered trees with
fixed/flexible {\em uneven} sub-wedges). For notational convenience,
we order all the sub-wedges $\{w_0(1), w_1(1), \cdots, w_0(n),
w_1(n)\}$ in Equation~(\ref{E-subWedge-uneven}) in ascending order
as
\begin{equation}\label{E-C3C4-subwedges}
m_{1}, m_{2}, ..., m_{n-1}, m_{n}, M_{n}, M_{n-1}, ..., M_{2},
M_{1}\nonumber
\end{equation}
where $m_{i}$ (resp., $M_{i}$) is the $i$-th minimum (resp.,
maximum) among all. That is, $c_i = (w_{t_i}(i), w_{t_i'}(i))$ for
$i = 1,...,n$ in Equation~(\ref{E-subWedge-uneven}) may be one of
the forms $(m_j, m_k)$, $(m_j, M_k)$, $(M_j, m_k)$, or $(M_j, M_k)$
for some $j, k \in \{1,...,n\}$. For convenience, each $m_{i}$
(resp., $M_{i}$) is said a type-$m$ (resp., type-$M$) sub-wedge.

For cases C3 and C4, we consider a bipartite graph $G = (V, U)$ and
a function $\phi: V \cup U \rightarrow \mathbb{R}$ in which
\begin{itemize}
  \item for case C3, $\phi(V) = \{w_{t_i}(i): i = 1, \cdots, n\}$, $\phi(U) = \{w_{t_i'}(i): i = 1, \cdots, n\}$;
  for case C4, $\phi(V) = \{M_1, \cdots, M_n\}$, $\phi(U) = \{m_1, \cdots, m_n\}$;

  \item the cost of each edge $(v,u)$ is $c(v,u) = \phi(v)+\phi(u)$ for RE, RA and DE problems; $c(v,u) = \phi(v) \times \phi(u)$ for SOP
  problem; the cost of a matching $N$ for $V \times U$ is
  $c(N) = \sum_{(v,u) \in N} c(v,u)$.
\end{itemize}
Note that, for convenience, each node in $V \cup U$
is also denoted by its $\phi$ function value.

In case C3 (unordered tree with fixed uneven sub-wedges), for each
$i = 1, 2, ..., n$, sub-wedge $w_{t_{i}}(i)$ in $V$ must be adjacent
to (matched with) sub-wedge $w_{t_j'}(j)$ for some $j \in \{1, 2,
..., n\}$ in $U$ in any solution of our concerned problems, and
hence the optimal solution must be a perfect matching $N$ for $V
\times U = \{w_{t_i}(i): i = 1, ..., n\} \times \{w_{t_i'}(i): i =
1, ..., n\}$.

In case C4 (unordered tree with flexible uneven sub-wedges), we have
the following observation.
\begin{observation}\label{obs:MmSolution}
For the RE4, RA4, DE4 or SOP4 problem, there must exist an optimal
solution in which each type-$m$ sub-wedge is adjacent to (matched
with) a certain type-$M$
sub-wedge.
\end{observation}
The above observation must hold; otherwise, there must exist $k$
pairs of adjacent type-$m$ sub-wedges and $k$ pairs of adjacent
type-$M$ sub-wedges for some $k \geq 1$ in the optimal drawing $D$.
But one can easily verify that any of our concerned aesthetic
criteria of drawing $D$ must be no better than the drawing where
each of the $2k$ type-$m$ sub-wedges is altered to be adjacent to a
certain of the $2k$ type-$M$ sub-wedges in drawing $D$ (i.e., a
drawing in Observation~\ref{obs:MmSolution}). Such an optimal
solution in Observation~\ref{obs:MmSolution} must be a perfect
matching $N$ for $V \times U = \{M_1, ..., M_n\} \times \{m_1, ...,
m_n\}$.

\begin{figure}[tbp]
\centering \scalebox{0.625}{\includegraphics{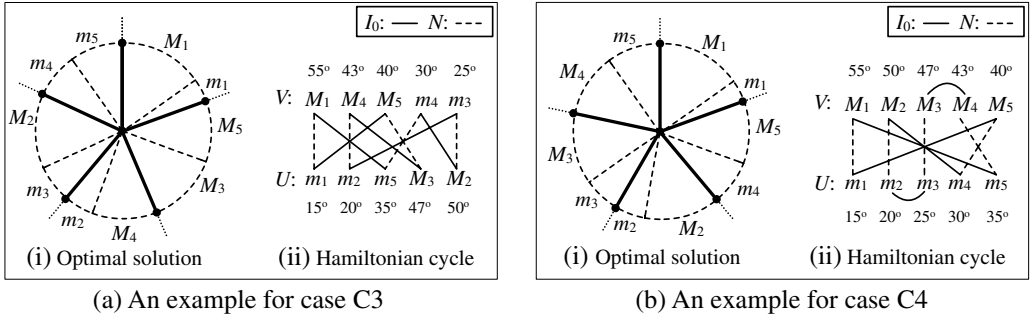}}%
\centering \caption{Two examples for expressing optimal solutions as
bipartite graphs.} \label{fgBipartite}
\end{figure}

If $I_0$ denotes the set of the edges corresponding to each pair
$(w_{t_i}(i),w_{t_i'}(i))$ for $i \in \{1,...,n\}$ (note that
$(w_{t_i}(i),w_{t_i'}(i)) \in$ $V \times U$ in case C3;
$(w_{t_i}(i),w_{t_i'}(i)) \in$ $V \times V \cup V \times U \cup U
\times V \cup U \times U$ in case C4), then $I_0 \cup N$ forms a
Hamiltonian cycle for $V \cup U$. 
Two examples for the same problem instance but under different cases
are shown in Figure~\ref{fgBipartite}, where the edges in $N$
(resp., $I_0$) are represented by dash (resp., solid) lines. As a
result, the RE (resp., RA; DE) problem is equivalent to finding a
matching $N_\opt$ for $V
\times U$ such that $I_0 \cup N_\opt$ is a 
Hamiltonian cycle of $V \cup U$ and the smallest edge cost in
$N_\opt$ is maximal (resp., the ratio of the largest and the
smallest edge costs in $N_\opt$ is minimal; the standard deviation
of the edge costs in $N_\opt$ is minimal).

Before showing our results, we introduce some notation as follows.
We place all the nodes in $V$ (resp., $U$) on the line $y=1$ (resp.,
$y=0$) of the $xy$-plane. 
Given any matching $N$ with two edges $e_1 = (v_a, u_b)$ and $e_2 =
(v_c, u_d)$ in $V \times U$, an {\em exchange} on $e_1$ and $e_2$
returns a matching $N'$ such that $N'=N\otimes (e_1, e_2) = (N
\setminus\{e_1, e_2\}) \cup \{(v_a,u_d),(v_c,u_b)\}$. Denote by
$e_v$ the edge incident to node $v$ in $N$.

\begin{thm}
\label{thm-RE3}\label{thm-RE4}%
The RE3 and RE4 problems can be solved in $O(n\log n)$ time.
\end{thm}
\begin{proof}
(Sketch) First consider the RE3 problem. A careful examination
reveals that the RE3 problem and the 2SAL problem are rather similar
in nature. Hence, Algorithm~\ref{alg:RE3-RE4} (a slight modification
of
the algorithm for the 2SAL) \cite{LV1997}
is sufficient to solve the RE3 problem in  $O(n\log n)$ time. 

\begin{algorithm}
\floatname{algorithm}{Algorithm} \caption{$\mbox{\sc
OptBalloonDrawing-RE3-RE4 }$} \label{alg:RE3-RE4}
\begin{algorithmic}[1]
\STATE  construct a bipartite graph $V \times U = \{w_{t_i}(i): i
=1, 2, ..., n\} \times \{w_{t_i'}(i): i = 1, 2, ..., n\}$ for RE3
        (resp., $V \times U = \{M_1, M_2, ..., M_n\} \times \{m_1, m_2, ..., m_n\}$ for RE4)\\
\STATE  sort the sizes of the sub-wedges in $V$ 
        in nonincreasing order as $\beta_1, \beta_2, ..., \beta_n$\\
\STATE  sort the sizes of the sub-wedges in $U$ 
        in nondecreasing order as $\alpha_1, \alpha_2, ..., \alpha_n$\\
\STATE  consider a matching $N$ in which $\alpha_i$ is matched with $\beta_i$ for each $i \in \{ 1, 2, ..., n\}$.
\IF     {$I_0 \cup N$ is a 
        Hamiltonian cycle for $V \cup U$}
\STATE  STOP
\ENDIF
\STATE order $\Omega = \{ \alpha_i + \beta_{i+1}: i = 1, 2, ..., n-1
\}$,
in nonincreasing order\\
\STATE $i \leftarrow 0$\\
\REPEAT
\STATE            $i \leftarrow i+1$\\
\IF{ $\alpha_j$ and $\beta_{j+1}$ belong to different cycles in $I_0
\cup N$, where $\alpha_j+\beta_{j+1}$ is the $i$-th maximum in $\Omega$}%
    \STATE $N \leftarrow N \otimes (e_{\alpha_j}, e_{\beta_{j+1}})$
\ENDIF%
\UNTIL{$I_0 \cup N$ is a 
Hamiltonian cycle for $V \cup
U$}
\end{algorithmic}
\end{algorithm}

The reader is referred to \cite{LV1997} for more details on the proof
of the correctness of the algorithm. A brief explanation for the correctness is given as follows.
From \cite{LV1997}, we have the following proposition and property:

\bigskip

\noindent {\bf Proposition 1.} {\em A matching $N$ determines a solution for RE3 if $I_0 \cup N$ is a unique cycle.}

\bigskip

\noindent {\bf Property 1.} {\em Let $optAngResl$ be the optimal solution for RE3.
Then $optAngResl \leq \min \{ \beta_i + \alpha_i , 1 \leq i \leq n \}$,
where $V = \{\beta_1, \cdots, \beta_n\}$; $U = \{\alpha_1, \cdots, \alpha_n\}$; $\beta_1 \geq \cdots \geq \beta_n$; $\alpha_1 \leq \cdots \leq \alpha_n$.}

\bigskip

See Algorithm~\ref{alg:RE3-RE4}. If $I_0 \cup N$ is a unique cycle at the end of Line~7, then Proposition 1 and Property 1 implies optimality;
otherwise, Lines~8--15 are executed.
At each iteration of the loop in Lines~10--15, no matter whether $N \leftarrow N \otimes (e_{\alpha_j}, e_{\beta_{j+1}})$ is executed or not,
the cases discussed in \cite{LV1997} can be tailored to show that the cost of each matched edge in $N$ is no less than $ optAngResl$.
Hence, the solution produced by Algorithm~\ref{alg:RE3-RE4} must be no less than $optAngResl$.

The time complexity of the algorithm is explained briefly as follows. It is easy to see that Lines~1--8 can be executed in $O(n\log n)$ time.
At the end of Line~7, the nodes of each various cycle are stored in a linked list in $O(n)$ time.
Let $\mathbb{S}$ be a stack storing the labels $\alpha_i$ top to bottom, in nonincreasing order of $\alpha_i + \beta_{i+1}$.
Stack $\mathbb{S}$ is used to detect which two cycles we merge next.
This is done by checking if the endpoints of the edge $(\alpha_i, \beta_{i+1})$, corresponding to top element $\beta_i$ of stack $\mathbb{S}$, belong to different cycles.
If they do, the two cycles are merged next;
otherwise, the element at the top of the stack is discarded.
Therefore, it takes $O(n)$ time to detect which cycles to merge.
The exchanging operation in Line~13 is done in $O(1)$ time.
But also, merging two cycles is equivalent to merging two linked lists, which is done in $O(1)$ time as well.
As a result, the time complexity of Algorithm~\ref{alg:RE3-RE4} is $O(n\log n)$.


In what follows, we consider the RE4 problem.
By Observation~\ref{obs:MmSolution}, we find an optimal solution for
the RE4 problem where each type-$m$ sub-wedge is adjacent to a
certain type-$M$ sub-wedge, i.e., a perfect matching $N$ for
$V\times U = \{M_1, M_2, ..., M_n\} \times \{m_1, m_2, ..., m_n\}$.
By viewing $m_i$ (resp., $M_i$) as $\alpha_i$ (resp., $\beta_i$) for
each $i \in \{1, ..., n\}$, the RE4 problem is similar to the RE3
problem. As a result, Algorithm~\ref{alg:RE3-RE4} can also be
applied to
solving the RE4 problem in $O(n\log n)$ time. 
\end{proof}


We  now turn our attention to the RA3 and RA4 problems. 
We consider a decision version of the RA3 (resp., RA4) problem:

\bigskip\noindent {{\sc The RA3 (resp., RA4) Decision Problem.}}\\
Given a balloon drawing of an unordered tree with fixed (resp.,
flexible) uneven sub-wedges, does there exist a circular permutation
$\sigma$ of $\{1, ..., n\}$ (resp., a circular permutation $\sigma$
of $\{1, ..., n\}$ and a sub-wedge assignment $t$) so that the size
of each angle is between $A$ and $B$? If the answer returns yes,
then $AspRatio_{\sigma,t} \leq B/A$.
\bigskip



Taking advantage of  the analogy between  RA3 (RA4) and 2SLW, we are
able to show:
\begin{thm}
\label{thm-RA3}\label{thm-RA4}%
Both the RA3 and RA4 problems are NP-complete.
\end{thm}
\begin{proof}
(Sketch)
RA3 and 2SLW bear a certain degree of similarity.
Recall that given a set of $n$ jobs and a range $[LB, UB]$, the 2SLW problem decides wether a circular permutation exists such that the workforce requirement (i.e., the sum of the workforce requirements for two jobs respectively executed at two stations at the same time) for each time period is between $LB$ and $UB$.
Given a balloon drawing of an unordered tree with fixed uneven sub-wedges, the RA3 decision problem decides whether a circular permutation so that the size of each angle (i.e., the sum of two adjacent subwedges respectively from two various children) is between $A$ and $B$.
It is obvious that the decision version of the RA3 problem can be captured by
the 2SLW problem (and vice versa) in a straightforward way,
hence  NP-completeness follows.

As for the RA4 problem, since the upper bound (i.e., in NP) for the
RA4 problem is easy to show, we show the RA4 problem to be NP-hard
by the reduction from the 2SLW problem as follows.

The idea of our proof is to design an RA4 instance so that one
cannot obtain any better solution by flipping sub-wedges. To this
end, from a 2SLW instance -- a set $\mathbb{J} = \{J_1, J_2, ...,
J_n\}$ of jobs and two numbers $LB, UB$ where $J_i = (W_{i1},
W_{i2})$ for each $i \in \{1, ..., n\}$, we construct a RA4 instance
-- a set of sub-wedges $\{w_0(1), w_1(1), \cdots, w_0(n), w_1(n)\}$
and two numbers $A$ and $B$ in which we let $W_{max} = \max\{W_{11},
W_{12}, \cdots, W_{n1}, W_{n2}\}$ and $\rho=2\pi/\sum_{j=1}^n
(W_{j1}+W_{j2}+W_{max})$; $w_0(i) = W_{i1}\times\rho$ and $w_1(i) =
( W_{i2} + W_{max} ) \times \rho$ for each $i \in \{1, ..., n\}$; $A
= (LB+W_{max})\times\rho$ and $B = (UB+W_{max})\times\rho$.

Now we show that there exists a circular permutation $\langle
J_{\delta_1}, J_{\delta_2}, ..., J_{\delta_n}\rangle$ of
$\mathbb{J}$ so that the workforce requirement for each time period
is between $LB$ and $UB$ if and only if there exist a circular
permutation $\sigma$ of $\{1, ..., n\}$ and a sub-wedge assignment
$t$ so that the size of each angle in the RA4 instance is between
$A$ and $B$.

We are given a 2SLW instance with a circular permutation $\langle J_{\delta_1}, J_{\delta_2},
..., J_{\delta_n} \rangle$ of $\mathbb{J}$ so that the workforce
requirement for each time period is between $LB$ and $UB$. It turns
out that $LB \leq W_{\delta_i2}+W_{\delta_{i\oplus 1}1} \leq UB$ for
each $i \in \{1, ..., n\}$. It implies that $(LB+W_{max})\times\rho
\leq (W_{\delta_i2} + W_{\delta_{i\oplus1}1} + W_{max})\times \rho
\leq (UB + W_{max})\times\rho$ for each $i \in \{1, ..., n\}$.
Consider $\sigma = \delta$ and $t = (0, 0, ..., 0)$ in the RA4
instance constructed above. Since $w_0(\sigma_i) =
W_{\sigma_i1}\times\rho$ and $w_1(\sigma_i) = ( W_{\sigma_i2} +
W_{max} ) \times \rho$ for each $i \in \{1, ..., n\}$ in the
construction, thus $(LB+W_{max})\times\rho \leq
w_1(\sigma_i)+w_0(\sigma_{i\oplus 1}) \leq (UB +
W_{max})\times\rho$. That is, $A \leq \theta_{\sigma_i} \leq B$ for
each $i \in \{1, ..., n\}$.

Conversely, we are given a
RA4 instance with a circular permutation $\sigma$ of $\{1, ..., n\}$
and a sub-wedge assignment $t$ so that the size of each angle in the
RA4 instance is between $A$ and $B$. For any $i,j \in \{1, ...,
n\}$, since $w_1(i) = (W_{i2} + W_{max})\times \rho \geq
W_{max}\times \rho \geq W_{j1}\times \rho = w_0(j)$, hence $w_1(i)
\geq w_0(j)$. In the RA4 instance, the size of each angle can be
$w_0(i)+w_0(j)$, $w_0(i)+w_1(j)$, or $w_1(i)+w_1(j)$ for some $i,j
\in \{1, ..., n\}$. For convenience, the angle with size
$w_0(i)+w_0(j)$ (resp., $w_0(i)+w_1(j)$; $w_1(i)+w_1(j)$) for some
$i,j \in \{1, ..., n\}$ is called a type-00 (resp., 01; 11) angle
(note that the order of $i$ and $j$ is not crucial here).

If there exists a type-00 angle in the RA4 instance, then there must
exist at least one type-11 angle in this instance; otherwise, all
the angles are type-01 angles.

In the case when there exists a type-00 angle with size
$w_0(i)+w_0(j)$ so that there exists a type-11 angle with size
$w_1(k)+w_1(l)$  for some $i,j, k, l \in \{1, ..., n\}$, then
w.l.o.g., the sub-wedge sequence of the instance is expressed as a
circular permutation $\langle S_1, w_0(i), w_0(j), S_2, w_1(k),
w_1(l), S_3 \rangle$ where $S_1$ -- $S_3$ are sub-wedge
subsequences; the number of sub-wedges in each of $S_1$ and $S_3$
(resp., $S_2$) is odd (resp.,
even)
. Let $S_2^R$ be the reverse of $S_2$. Consider a new circular
permutation $\langle S_1, w_0(i), w_1(k), S_2^R, w_0(j), w_1(l), S_3
\rangle$, in which the size of each angle is between $A$ and $B$,
because the size of each angle in $S_3 \cup S_1$ and $S_2^R$ is
originally between $A$ and $B$; $A \leq w_0(i)+w_1(k) \leq B$ (since
$w_0(i)+w_1(k) \geq w_0(i)+w_0(j) \geq A$ and $w_0(i)+w_1(k) \leq
w_1(l)+w_1(k) \leq B$); similarly, $A \leq w_0(j)+w_1(l) \leq B$.

If there still exists a type-00 angle in the new circular
permutation, then we repeat the above procedure until we obtain a
circular permutation $\delta$ where all the angles are type-01
angles. By doing this, the size of each angle in $\delta$ is between
$A$ and $B$, and the sub-wedge assignment $t$ in the drawing
achieved by $\delta$ is $(0,0,...,0)$ or $(1,1, ..., 1)$. In the
case of $t = (1, 1, ..., 1)$, we let $\delta \leftarrow \delta^R$,
then $t$ becomes $(0,0,...0)$.

Consider the 2SLW instance (constructed above) corresponding to the
circular permutation $\delta$. In the 2SLW instance, for each $i \in
\{1, ..., n\}$, workforce requirement
$W_{\delta_i2}+W_{\delta_{i\oplus1}1} = (w_1(\delta_i) +
w_0(\delta_{i\oplus1}))/\rho - W_{max}$. Hence, $A/\rho -W_{max}
\leq W_{\delta_i2}+W_{\delta_{i\oplus1}1} \leq B/\rho - W_{max}$,
which implies $LB \leq W_{\delta_i2}+W_{\delta_{i\oplus1}1} \leq
UB$.
\end{proof}

We can utilize a technique similar to the reduction from {\em
Hamiltonian-circle problem on cubic graphs} (HC-CG) to 2SLW
(\citep{V2003}) to establish NP-hardness for DE3 and DE4. Hence, we
have the following theorem, whose proof is given in Appendix 
because it is too cumbersome and our main result for the DE3 and DE4
problems is to design their approximation algorithms.

\begin{thm}
\label{thm-DE3}\label{thm-DE4}%
Both the DE3 and DE4 problems are NP-complete.
\end{thm}

\section{Approximation Algorithms for Those Intractable Problems}
\label{sec:approx}

We have shown RA3 and RA4 to be NP-complete. The results on
approximation algorithms for those problems are given as follows.

\begin{thm}
\label{thm-RA3-approx}\label{thm-RA4-approx}%
Algorithm~$\ref{alg:RE3-RE4}$ is a $2$-approximation algorithm for
RA3 and RA4.
\end{thm}
\begin{proof}
Let $a_\angResl$ (resp., $b_\angResl$ and $r_\angResl$) be the
minimal angle (resp., the maximal angle and the aspect ratio) among
the circular permutation generated by 
Algorithm~\ref{alg:RE3-RE4}. 
Denote $a_\opt$ (resp., $b_\opt$ and $r_\opt$) as the maximum of the
minimal angle (resp., the minimum of the maximal angle and the
optimal aspect ratio) among any circular permutation. Since
$b_\angResl \leq 2 M_1 \leq 2 (x+M_1) \leq 2 b_\opt$ where $x$ is
the sub-wedge adjacent to $M_1$ in the circular permutation with the
minimum of the maximal angle, we have $b_\angResl \leq 2 b_\opt$. By
Theorem~\ref{thm-RE3}, we have $a_\angResl = a_\opt = optAngResl$.
Therefore, $r_\angResl = b_\angResl / a_\angResl \leq 2 b_\opt /
a_\opt \leq 2 r_\opt$.
\end{proof}

Next, we design approximation algorithms for the NP-complete DE
problems. Here we only consider the approximation algorithms for the
SOP4 and DE4 problems because the approximation algorithms for the
SOP3 and DE3 problems are similar and simpler. Recall that the SOP4
problem is equivalent to finding a matching $N_\opt$ for bipartite
graph $V \times U$, such that $c(N_\opt)$ is the minimal, where
$c(N) = \sum_{(v,u)\in N} \phi(v) \times \phi(u)$.

Consider a matching $N_D$ for bipartite graph $V \times U$ in which
$M_i$ is matched with $m_i$ for each $i$, i.e.,
$c(N_D)=\sum_{i=1}^n M_i m_i.
$ 
 Assume that $I_0 \cup N_D$ consists of $\eta$ subcycles for $1 \leq
\eta \leq n$, in which we recall that $I_0$ denotes the set of the
edges corresponding to each pair $(w_{t_i}(i), w_{t_i'}(i))$ for $i
\in \{1, \cdots, n\}$. According to matching $N_D$, we have that
each subcycle in $I_0 \cup N_D$ contains at least one matched edge
between $M_i$ and $m_i$ for some
$i$. Let the %
{\em exchange graph} $\chi = (V_\chi, E_\chi)$ for bipartite graph
$V \times U$ be a complete graph in which
\begin{itemize}
  \item each node in $V_\chi$ corresponds to a subcycle of $I_0 \cup N_D$,
i.e., $|V_\chi| = \eta$;
  \item each edge $e_i = (u,v)$ in $E_\chi$
corresponding to two subcycles $C_u$ and $C_v$ in $I_0 \cup N_D$ has
cost $\psi(e_i) = \min\{r_{a,b}s_{b,a}| (M_a, m_b) \in (C_u,
C_v)\cup (C_u, C_v) \mbox{ for any $a$, $b$; }$ $ r_{a,b} = M_a -
M_b$, $s_{b,a} = m_b - m_a\}$. (In fact, the cost represents the
least cost of exchanging edges $e_{M_a}$ and $e_{m_b}$ in $V \times
U$.)
\end{itemize}
When $\psi(e_i) = r_{k,l}s_{l,k}$ for some $k, l$, we denote
$\mu(e_i) = k$ and $\nu(e_i) = l$. Let $T_\chi = (V_\chi,
E_{T_\chi})$ be a minimum spanning tree over $\chi$. With exchange
graph $\chi$ and its minimum spanning tree $T_\chi$ as the input of
Algorithm~\ref{alg:SOP4}, we can show that Algorithm~\ref{alg:SOP4}
is a 2-approximation algorithm for the SOP4 problem.

\begin{algorithm}[tbp]{\small
\floatname{algorithm}{Algorithm} \caption{$\mbox{\sc
ApproxBalloonDrawing-SOP4}$}\label{alg:SOP4}
\begin{algorithmic}[1]
\STATE construct the exchange graph $\chi = (V_\chi, E_\chi)$ for $V\times U$
\STATE find the minimum spanning tree $T_\chi = (V_\chi, E_{T_\chi})$ of exchange graph $\chi$ where $|V_\chi| = \eta$
\STATE let $S_i=\{ M_{\mu(e_i)}, m_{\mu(e_i)}, M_{\nu(e_i)},
m_{\nu(e_i)} \}$ for each edge $e_i \in E_{T_\chi}$ (noticing that if $\psi(e_i) = r_{k,l}s_{l,k}$ for some $k, l$, then $\mu(e_i) = k$ and $\nu(e_i) = l$), where each $e_i$ is said to correspond to $S_i$ (i.e., there are $S_1, S_2, \cdots, S_{\eta-1}$)
\STATE let $S = \{S_1, \cdots, S_{\eta-1}\}$
\FOR {{\bf each} set $S_a$ in $S$}
    \FOR {{\bf each} element $x$ in $S_a$}
        \STATE find a set $S_b$ that includes element $x$ but is not considered before
        \STATE append the elements in set $S_b$ to the end of set $S_a$ (i.e., the duplicate elements are not deleted)
        \STATE let both edges $e_i$ and $e_j$ correspond to $S_a$, where edges $e_i$ and $e_j$ in $T_\chi$ correspond to $S_a$ and $S_b$, respectively
        \STATE $S \leftarrow S \setminus S_b$
    \ENDFOR
\ENDFOR
\STATE for each set in $S$, remove the duplicate elements in each set
\STATE order the elements in each set $S_i$, and then denote the new set as
$S_i' = \{m_1', m_2', \cdots, m_l', M_l', M_{l-1}', \cdots, M_1'\}$
where $m_i'$ (resp.,
$M_i'$) is the $i$-th minimum (resp., maximum) in $S_i$; the cardinality of $S_i'$ is $2 l$\\
\FOR {{\bf each} $S_i'$}
    \STATE $M_j'$ is matched with $m_{j + 1}'$ for $j = 1, \cdots, l-1$
    \STATE $M_l'$ is matched with $m_1'$
\ENDFOR
\STATE output such a matching $N_\APX$ for $V \times U$
\end{algorithmic}
}\end{algorithm}

\begin{figure}[tbp]
\centering \scalebox{0.625}{\includegraphics{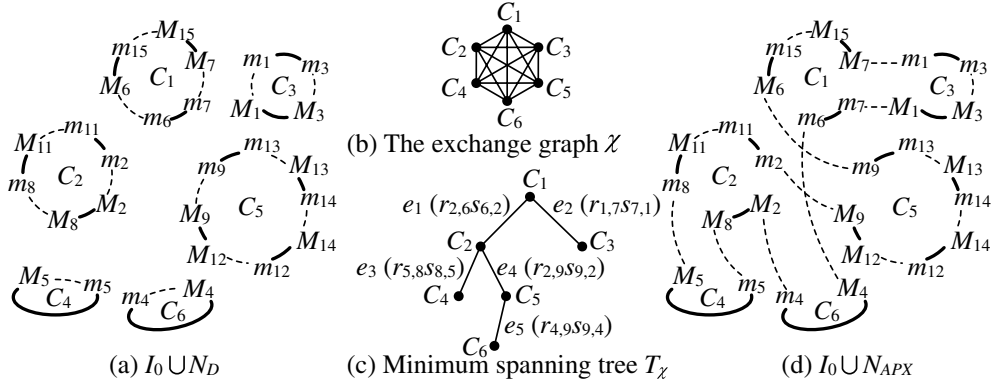}}%
\centering \caption{An example showing how Algorithm~\ref{alg:SOP4}
works.} \label{fgAlgSOP4}
\end{figure}

Figure~\ref{fgAlgSOP4} gives an example to illustrate how the
algorithm works. Figure~\ref{fgAlgSOP4}(a) is $I_0 \cup N_D$ where
the solid lines (resp., dash lines) are the edges in $I_0$ (resp.,
in $N_D$). Figure~\ref{fgAlgSOP4}(b) is its exchange graph $\chi$,
and we assume that Figure~\ref{fgAlgSOP4}(c) is the minimum spanning
tree $T_\chi$ for $\chi$ where each edge $e_i$ in $T_\chi$ has
weight $r_{\mu(e_i),\nu(e_i)}s_{\nu(e_i),\mu(e_i)}$. We illustrate
each $S_i$ after each modification in Line~11 of
Algorithm~\ref{alg:SOP4} as follows:
\begin{itemize}
    \item Initial:
    $S_1 = \{M_2, m_2, M_6, m_6\}$, $S_2 = \{M_1, m_1, M_7, m_7\}$, $S_3 = \{M_5, m_5, M_8, m_8\}$, $S_4 = \{M_2, m_2, M_9, m_9\}$, $S_5 = \{M_4, m_4, M_9, m_9\}$.%
    \item The elements in $S_4$ is appended to the end of $S_1$:\\
    $S_1 = \{M_2, m_2, M_6, m_6, M_2, m_2, M_9, m_9\}$, $S_2 = \{M_1, m_1, M_7, m_7\}$, $S_3 = \{M_5, m_5,$ $M_8, m_8\}$, $S_5 = \{M_4, m_4, M_9, m_9\}$.%
    \item The elements in $S_5$ is appended to the end of $S_1$:\\
    $S_1 = \{M_2, m_2, M_6, m_6, M_2, m_2, M_9, m_9, M_4, m_4, M_9, m_9\}, S_2 = \{M_1, m_1, M_7, m_7\},$ $S_3 = \{M_5, m_5, M_8, m_8\}$.%
\end{itemize}
Based on the above, Algorithm~\ref{alg:SOP4} returns $N_\APX$, and
$I_0 \cup N_\APX$ is shown in Figure~\ref{fgAlgSOP4}(d). In fact,
Algorithm~\ref{alg:SOP4} provides a 2-aproximation algorithm for
SOP4. A slight modification also yields a 2-approximation algorithm
for SOP3.


Before showing our result, we need the following notation and lemma.
A {\em permutation} $\pi$ is a 1-to-1 mapping of $\{1, ..., n\}$
onto itself, which can be expressed as: $\pi = (\pi(1), \pi(2), ...,
\pi(n))$
or in compact form in terms of {\em factors}. (Note that it is
different from the circular permutation used previously.) If
$\pi(j_k)=j_{k+1}$ for $k=1,2,...,h-1$, and $\pi(j_h)=j_1$, then
$\langle j_1, j_2, ..., j_h \rangle$ is called a {\em factor} of the
permutation $\pi$. A factor with $h \geq 2$ is called a nontrivial
factor. Note that a matching $N$ for the bipartite graph $V \times
U$ constructed above can be viewed as a permutation $\pi: V
\rightarrow U$.

\begin{lemma}\label{lemma-perm}
    For $n \geq 2$, let $X = \{x_1, x_2, \cdots, x_n\}$ (resp., $Y = \{y_1, y_2, \cdots, y_n\}$)
    where $x_i$ (resp., $y_i$) is the $i$-th maximum (resp., minimum) among all. Let $\varrho: X
    \rightarrow Y$ be a $1$-to-$1$ mapping, i.e., a permutation of
    $\{1,\cdots,n\}$. 
    If $\varrho(X)$ is a permutation consisting of only a
    nontrivial factor with size $n$, then
    \begin{equation} \label{E-lemma-perm-1}
    c(\varrho(X)) = \sum_{i=1}^n x_i y_{\varrho(i)} \geq
    \sum_{i=1}^n x_i y_i + \sum _{i=1}^{n-1} r_{i,i+1}s_{i+1,i}
    \end{equation}
    where $r_{a,b} = x_a - x_b, s_{c,d} = y_c - y_d$ for any $a,b,c,d$. Moreover, if
    $r_{j,i+1}s_{i+1,j'} - r_{i,i+1} s_{i+1,i} \geq 1$ for each $i,j,j' \in \{1, \cdots, n-1\}$ and $j,j' < i$,
    then
    \begin{equation} \label{E-lemma-perm-2}
    c(\varrho(X)) \geq \sum_{i=1}^n x_i y_i + \sum _{i=1}^{n-1}
    r_{i,i+1}s_{i+1,i} + n - 2
    \end{equation}
    Note that the difference between
    Equation~$(\ref{E-lemma-perm-1})$
    and Inequality~$(\ref{E-lemma-perm-2})$ is that Inequality~$(\ref{E-lemma-perm-2})$ can be applied only when the factor size $n$ is
    known.
\end{lemma}
\begin{proof} We proceed by induction on the size of $\varrho(X)$. If $n=2$,
$c(\varrho(X))-\sum_{i=1}^2 x_i y_i = x_1 y_2 + x_2 y_1 - x_1 y_1 -
x_2 y_2 = r_{1,2} s_{2,1}$ holds. Suppose that the required two
inequalities hold when $n=k$. When $n=k+1$,
\begin{eqnarray}
c(\varrho(X)) &=& \sum_{i \in \{1,\cdots,k\}\setminus\{\varrho^{-1}(k+1)\}} x_i y_{\varrho(i)} + x_{\varrho^{-1}(k+1)}y_{k+1} + x_{k+1}y_{\varrho(k+1)} %
\nonumber\\
&=& \sum_{i=1}^{k} x_i y_{\varrho'(i)} + x_{\varrho^{-1}(k+1)}y_{k+1} + x_{k+1}y_{\varrho(k+1)} - x_{\varrho^{-1}(k+1)}y_{\varrho(k+1)} %
\nonumber
\end{eqnarray}
where $\varrho'$ is a size-$k$ permutation consisting of a nontrivial
factor with size $k$. Then,
\begin{eqnarray}
c(\varrho(X)) &=& \sum_{i=1}^{k} x_i y_{\varrho'(i)} + x_{k+1} y_{k+1}
+ ( x_{\varrho^{-1}(k+1)} - x_{k+1} )( y_{k+1} - y_{\varrho(k+1)}
)\nonumber\\
&=& \sum_{i=1}^{k} x_i y_{\varrho'(i)} + x_{k+1} y_{k+1} +
r_{\varrho^{-1}(k+1),k+1} s_{k+1,\varrho(k+1)}\label{E-lemma-perm-3}
\end{eqnarray}

For proving Equation~(\ref{E-lemma-perm-1}), we replace the first
term in Equation~(\ref{E-lemma-perm-3}) by the inductive hypothesis
of Equation~(\ref{E-lemma-perm-1}), and then obtain:
\begin{eqnarray}
c(\varrho(X)) &\geq& \sum_{i=1}^{k+1} x_i y_i + \sum_{i=1}^{k-1}
r_{i,i+1} s_{i+1,i} + r_{\varrho^{-1}(k+1),k+1} s_{k+1,\varrho(k+1)} \nonumber \\
&\geq& \sum_{i=1}^{k+1} x_i y_i + \sum_{i=1}^{k} r_{i,i+1}
s_{i+1,i}\nonumber
\end{eqnarray}
since $x_{\varrho^{-1}(k+1)} \geq x_k$ and $y_{\varrho(k+1)} \leq
y_k$.

For proving Equation~(\ref{E-lemma-perm-2}), we replace the first
term in Equation~(\ref{E-lemma-perm-3}) by the inductive hypothesis
of Equation~(\ref{E-lemma-perm-2}), and then obtain:
\begin{eqnarray}
c(\varrho(X)) &\geq& \sum_{i=1}^{k+1} x_i y_i + \sum_{i=1}^{k-1}
r_{i,i+1} s_{i+1,i} + k-2 + r_{\varrho^{-1}(k+1),k+1} s_{k+1,\varrho(k+1)}\nonumber\\
&\geq& \sum_{i=1}^{k+1} x_i y_i + \sum_{i=1}^{k} r_{i,i+1} s_{i+1,i}
+ k-1\nonumber
\end{eqnarray}
since $( x_{\varrho^{-1}(k+1)} - x_{k+1} )( y_{k+1} - y_{\varrho(k+1)}
) \geq (x_k - x_{k+1})(y_{k+1} - y_k) + 1$ by the premise of
Equation~(\ref{E-lemma-perm-2}) (Note that the permutation consists
of a nontrivial factor of size $n$, and hence the case
$\varrho^{-1}(k+1)=\varrho(k+1)=k$ does not occur except for $n=2$).
\end{proof}

Now, we are ready to show our result:

\begin{thm}
\label{thm-DE4-approx1}%
There exist $2$-approximation algorithms for SOP3 and SOP4, which
run in $O(n^2)$ time.
\end{thm}
\begin{proof}
Recall that given an unordered tree with fixed (resp., flexible) subwedges, the SOP3 (resp., SOP4) problem is to find a circular permutation $\sigma$ of $\{1, \cdots, n\}$ (resp., a circular permutation $\sigma$ of $\{1, \cdots, n\}$ and a sub-wedge assignment $t$) so that the sum of products of adjacent subwedge sizes ($SOP_{\sigma,t}$) is as small as possible.
We only consider SOP4; the proof of SOP3 is similar and simpler.
In what follows, we show that Algorithm~\ref{alg:SOP4} correctly produces the 2-approximation solution for SOP4 in $O(n \log n)$ time. 

From \cite{KS2009}, we have $c(N_\opt) \geq c(N_D)$, which is explained briefly as follows.
From \citep{KS2009}, we have that $N_D$ can be transformed from $N_\opt$
by a sequence of exchanges $x_1, x_2, \cdots, x_n$ which can be
constructed as follows. Let $N_k$ denote the matching transformed by
the sequence of exchanges $x_1, x_2, \cdots, x_k$ for $k \leq n$. We
say a node $v$ in $V$ is {\em satisfied} in $N_k$ if its adjacent
node in $N_k$ is the same as its adjacent node in $N_D$. For $i = 1,
2, \cdots, n$, if the sub-wedge $M_i$ is satisfied, then $x_i$ is a
null exchange. Otherwise, if the node adjacent to $M_i$ in $N_i$ is
adjacent to the sub-wedge $M_j$ in $N_\opt$ for $i \neq j$ (i.e.,
$M_i$ is not adjacent to $m_i$ in $N_i$), then let $x_i$ be the
exchange between the edges respectively incident to
$M_i$ and $M_j$ in $N_i$. 
Here, by observing each non-null exchange $x_i$, $\phi(N_\opt) -
\phi(N_i) = r_{i,j}s_{j,i} \geq 0$. Hence, $\phi(N_\opt) \geq
\phi(N_n) = \phi(N_D)$.

Let
\begin{equation}
c_\LB = \sum_{i=1}^n M_i m_i + \sum_{e\in E_{T_\chi}}
r_{\mu(e),\nu(e)} s_{\nu(e),\mu(e)}.\nonumber
\end{equation}
We claim that $c(N_\opt) \geq c_\LB$. Since $I_0 \cup N_\opt$ is a
Hamiltonian cycle transformed from $I_0 \cup N_D$ consisting of
$\eta$ subcycles, there exist at least $\eta - 1$ times of merging
subcycles during the transformation (the sequence of exchanges). We
can view $N_\opt$ as a permutation with several factors
. There must exist a set $\Lambda$ of $\eta-1$ edges in $E_\chi$
forming a spanning tree for exchange graph $\chi$ such that each
edge in $\Lambda$ must correspond to an edge in $N_\opt$ which
cannot be in a trivial factor of permutation $N_\opt$, i.e., it
cannot be $M_i m_i$ for some $i$. Therefore, by
Inequality~(\ref{E-lemma-perm-1}) of Lemma~\ref{lemma-perm},
$c(N_\opt) \geq \sum_{i=1}^n M_i m_i + \sum_{e \in \Lambda}
r_{\mu(e),\nu(e)} s_{\nu(e),\mu(e)} \geq \sum_{i=1}^n M_i m_i +
\sum_{e\in E_{T_\chi}} r_{\mu(e),\nu(e)}s_{\nu(e),\mu(e)} = c_\LB$
since $E_{T_\chi}$ is the edge set of minimum spanning tree of
$\chi$.

In what follows, we show the approximation ratio to be 2. Note that
$N_\APX$ denotes the matching generated by Algorithm~\ref{alg:SOP4}.
Let $\mathbb{S}=\cup_{i=1}^{\eta-1} S_i$ and
$i(\mathbb{S})=\cup_{\forall e_i \in E_{T_\chi}}\{\mu(e_i),
\nu(e_i)\}$ in Algorithm~\ref{alg:SOP4}.
\begin{eqnarray}
&&2c(N_\opt) \geq 2 c_\LB 
\geq 2 \sum_{i=1}^n M_i m_i + \sum_{e\in E_{T_\chi}} r_{\mu(e),\nu(e)}s_{\nu(e),\mu(e)}\nonumber\\
&&\geq \sum_{e\in E_{T_\chi}} (M_{\mu(e)}m_{\mu(e)} +
M_{\nu(e)}m_{\nu(e)}) + \sum_{i \in \{1,2,\cdots,n\} \setminus
i(\mathbb{S}) } M_i m_i + \sum_{e\in E_{T_\chi}}
r_{\mu(e),\nu(e)}s_{\nu(e),\mu(e)}\nonumber
\end{eqnarray}
The last inequality above holds since $M_i m_i$ for any $i \in
i(\mathbb{S})$ never presents in the first summation term more than
twice; otherwise we can find another spanning tree with cost
strictly less than that of $T_\chi$. For example, we consider
Figure~\ref{fgAlgSOP4}(c). Suppose that the cost of edge $e_3$ in
$T_\chi$ is $r_{2,5}s_{5,2}$, rather than $r_{5,8}s_{8,5}$, i.e.,
$M_2 m_2$ is used three times by $e_1$, $e_3$, and $e_4$ (with costs
$r_{2,6}s_{6,2}$, $r_{2,5}s_{5,2}$, and $r_{2,9}s_{9,2}$,
respectively). We can obtain a contradiction by considering a
spanning tree $T$ replacing edge $e_4$ by edge $C_4 C_5$ with cost
$r_{5,9}s_{9,5}$, which is less than $r_{2,9}s_{9,2}$ in general.
(The cost of $T$ is less than that of $T_\chi$.)

Recall that $r_{a,b} = M_a - M_b$ and $s_{c,d} = m_c - m_d$. Hence,
combining the first and third terms of the above inequality, we
obtain:
\begin{eqnarray}
2c(N_\opt) &\geq& \sum_{e \in E_{T_\chi}} (M_{\mu(e)}m_{\nu(e)} +
M_{\nu(e)}m_{\mu(e)}) + \sum_{i \in \{1,2,\cdots,n\} \setminus
i(\mathbb{S}) } M_i m_i\nonumber\\
&\geq& \sum_{i=1}^{\eta-1} \sum_{j=1}^{|S_i'|-1} (M_j'm_{j + 1}' +
M_{j + 1}'m_j') + \sum_{i \in \{1,2,\cdots,n\} \setminus
i(\mathbb{S}) } M_i m_i\nonumber
\end{eqnarray}
The above inequality holds due to $\mu(e) \neq \nu(e)$ for any $e
\in E_\chi$. Since $M_2'm_1' \geq M_{|S_i'|}'m_1'$ in every $S_i'$,
we obtain:
\begin{eqnarray}
2c(N_\opt) &\geq&  \sum_{i=1}^{\eta-1} \left( \sum_{j=1}^{|S_i'|-1}
(M_j'm_{j + 1}') + M_{|S_i'|}'m_1'\right) + \sum_{i \in
\{1,2,\cdots,n\} \setminus i(\mathbb{S}) } M_i m_i =
c(N_\APX)\nonumber
\end{eqnarray}

In what follows, we explain how the algorithm runs in $O(n^2)$ time.

\bigskip In Line 1, the exchange graph can be constructed in $O(n^2)$ time as follows.
It takes $O(n^2)$ time to construct a complete graph $\chi$ with $\eta \leq n$ nodes in which the nodes corresponds $\eta$ subcycles in $I_0 \cup N$, and the cost of each edge is assumed to be infinity.
Then, it takes $O({{n}\choose{2}}) = O(n^2)$ time to compute all possible $r_{a,b} s_{b,a} = (M_a - M_b) (m_b - m_a)$ for any $a, b \in \{1, \cdots, n\}$. Consider each $r_{a,b} s_{b,a}$. If $M_a$ and $M_b$ belong to two different subcycles in $I_0 \cup N$, say $C_u$ and $C_v$, respectively, and $r_{a,b} s_{b,a} < \psi(e_i)$ for their corresponding edge $e_i = (u,v)$ in graph $\chi$, then $\psi(e_i) \leftarrow r_{a,b} s_{b,a}$.
Obviously, after considering all possible $r_{a,b} s_{b,a}$ in $O(n^2)$ time, graph $\chi$ is the required exchange graph.

In Line 2, it is well-known that the minimum spanning tree for graph $\chi$ can be found in $O(n \log n)$ time.
Line~3 runs in $O(n)$ time since each element is denoted only once. Line~4 is done in $O(n)$ time.

We explain how Lines~5--13 can be done in $O(n)$ time as follows.
Note that in Line~3, in addition that each set includes four elements, we record that each element knows which set includes it.
Hence, in Line~7, any set $S_b$ including element $x$ can be found in $O(1)$ time. Line~8 is done in $O(1)$ time, since each set is a linked list. Note that in Line~7 all the sets that includes element $x$ will be considered at the end of Line~12, because in Line~8 a duplicate element of $x$ is appended to $S_a$ and will be considered again in later iteration. Lines~9 and 10 are done in $O(1)$ time.
Therefore, Lines~7--10 are done in $O(1)$ time. We observe from Lines 5, 6, 8, 10 that each element in $S_1$, \dots, $S_{\eta-1}$ is considered once at the end of Line~12.
Since the number of elements in $S_1, \cdots, S_{\eta-1}$ is $4(\eta-1)$, there are $4(\eta-1)$ iterations, each of which is done in $O(1)$ time. Hence, Lines~5--12 are done in $O(4(\eta-1)) = O(n)$ time.
In Line~13, by scanning each set in $S$, all duplicate elements are deleted in $O(n)$ time.

Line 14 can be done in $O(n)$ time, because the ordering of $\{m_1, m_2, \cdots, m_n,$ $M_n, M_{n-1}, \cdots, M_1\}$ is known.
Lines~15--18 are done in $O(n)$ time, because each element is matched only once.
\end{proof}

Note that Algorithm~\ref{alg:SOP4} is a 2-approximation algorithm
for the SOP4 problem rather than the DE4 problem because the
approximation ratio is incorrect when the minus of the first and
third items inside the square root of Equation~(\ref{E-stdDev2}) is
negative. Therefore, we rewrite Equation~(\ref{E-stdDev2}) as:
{\footnotesize
\begin{eqnarray}
StdDev_{\sigma,t}%
= \sqrt{\frac{\sum_{i=1}^n(M_i^2 + m_i^2)}{n}%
+ \frac{2\sum_{i=1}^n w_{t_i'}(\sigma_i) w_{t_{i\oplus 1}'}(\sigma_{i\oplus 1})}{n}%
- \left(\frac{\sum_{i=1}^n (M_i + m_i)}{n}\right)^2}\nonumber\\
= \sqrt{\frac{\sum_{i=1}^n(M_i + m_i)^2}{n}%
+ \frac{-2\sum_{i=1}^n M_i m_i}{n} + \frac{2\sum_{i=1}^n w_{t_i'}(\sigma_i) w_{t_{i\oplus 1}'}(\sigma_{i\oplus 1})}{n}%
- \left(\frac{\sum_{i=1}^n (M_i+m_i)}{n}\right)^2}.\nonumber
\end{eqnarray}
}Note that the combination of first and fourth items inside the
square root of the above equation is the variance of $\{M_1+m_1,
M_2+m_2, \cdots, M_n+m_n\}$, and hence must be positive. Therefore,
the DE4 problem is equivalent to minimizing the sum of the second
and third items, i.e., to minimize
\begin{eqnarray}
\sum_{i=1}^n w_{t_i'}(\sigma_i) w_{t_{i\oplus 1}'}(\sigma_{i\oplus 1}) - \sum_{i=1}^n M_i m_i = SOP_{\sigma, t} - \sum_{i=1}^n M_i m_i.\nonumber%
\end{eqnarray}
Algorithm~\ref{alg:DE4} provides an $O(\sqrt{n})$-approximation
algorithm for DE4. A slight modification also yields an
$O(\sqrt{n})$-approximation algorithm for DE3.
Figure~\ref{fgAlgDE4}(a) is an example for Algorithm~\ref{alg:DE4}.

\begin{algorithm}[tbp]{\small
\floatname{algorithm}{Algorithm} \caption{$\mbox{\sc
ApproxBalloonDrawing-DE4}$ } \label{alg:DE4} %
The algorithm is almost the same as Algorithm~\ref{alg:SOP4} except
Lines~15--18 in Algorithm~\ref{alg:SOP4} is replaced as follows:
\smallskip\\
13': {\bf for each} $S_i'$ with $|S_i'| \geq 2$ (otherwise trivially) {\bf do}\\
14': \ \ \ let $r_{a,b}' = M_a'- M_b'$ and $s_{c,d}' = m_c' - m_d'$\\
15': \ \ \ an element is said to be {\em available} if it is not matched yet\\
16': \ \ \ {\bf for each} $j = 1, \cdots, |S_i'|-2$ {\bf do}\\
17': \ \ \ \ \ \ {\bf if} $r_{j,j+1}' \geq r_{j+1,j+2}'$ {\bf do}\\
18': \ \ \ \ \ \ \ \ \ the available maximum is matched with the available second minimum\\
19': \ \ \ \ \ \ {\bf else}\\
20': \ \ \ \ \ \ \ \ \ the available minimum is matched with the available second maximum\\
21': \ \ \ \ \ \ {\bf end if}\\
22': \ \ \ {\bf end for}\\
\vspace{-0.8cm}
\begin{quote}{\small
\hspace{-1.16cm} 23': \ \ $M_a'$ is matched with $m_{|S_i'|}'$; $M_{|S_i'|}'$ is matched with $m_b'$, where $M_a'$ and
$m_b'$ are the remaining elements excluded in the above condition for some $a, b \in \{1,\cdots,|S_i'|-1\}$\\
}
\end{quote}
\vspace{-0.8cm}
24': {\bf end for}
}\end{algorithm}

\begin{figure}[tbp]
\centering \scalebox{0.7}{\includegraphics{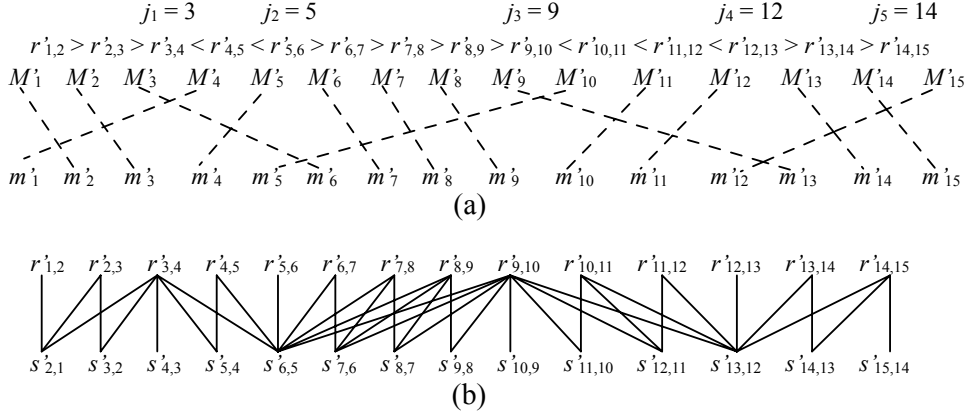}}%
\centering \caption{An example showing how Algorithm~\ref{alg:DE4}
works. (a) Certain $N_{S_i'}$ with $|S_i'|=15$ in $N_\APX$. (b)
Illustration of $c(N_{S_i'})-\sum_{j=1}^{|S_i'|}M_j'm_j'$ induced by
(a).} \label{fgAlgDE4}
\end{figure}

\begin{figure}[tbp]
\centering \scalebox{0.7}{\includegraphics{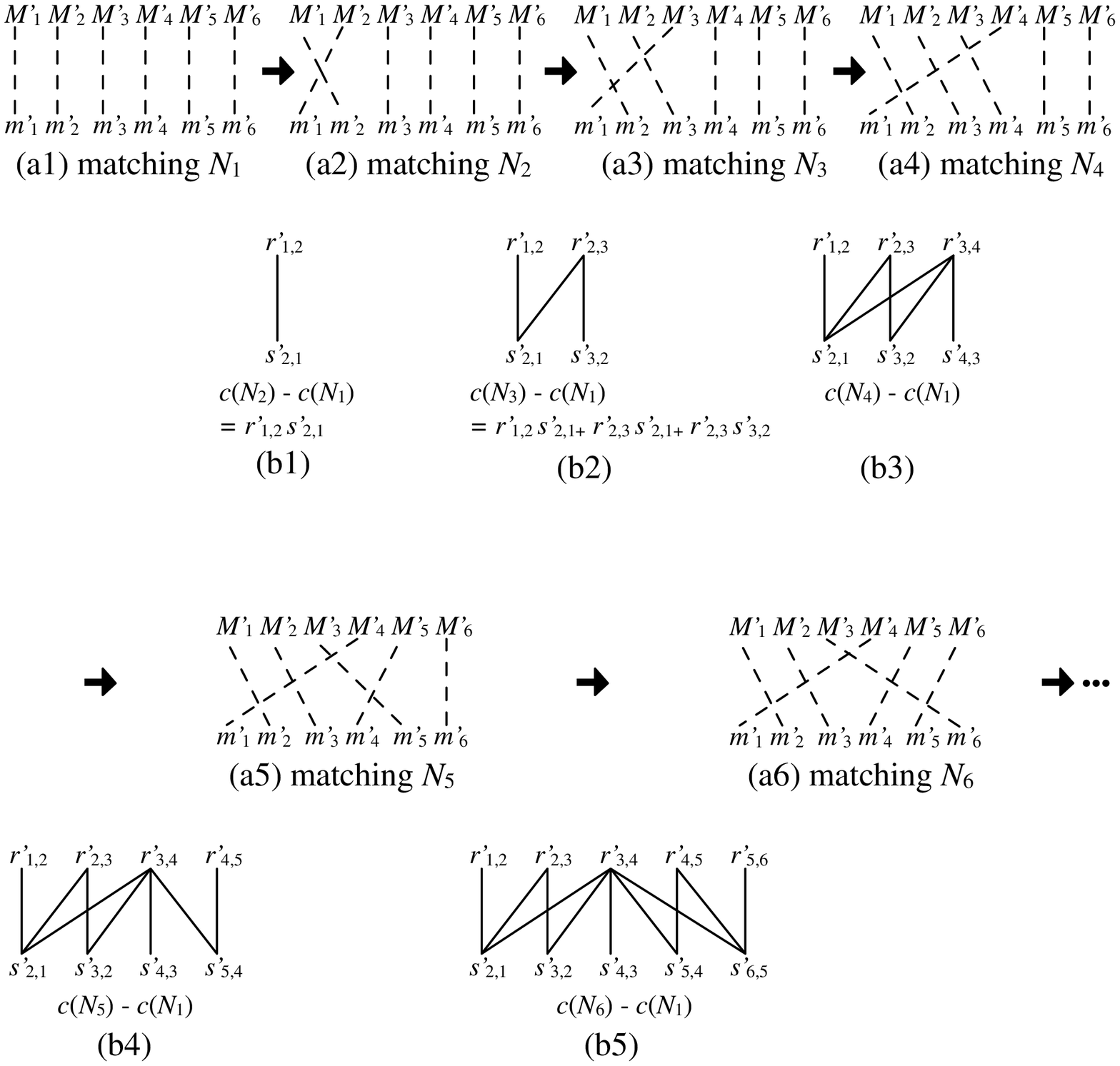}}%
\centering \caption{Illustration of the first several intermediate
steps of how to obtain Figure~\ref{fgAlgDE4}(b) from
Figure~\ref{fgAlgDE4}(a). For $i = 1, \cdots, 5$, matching $N_{i+1}$
is obtained by exchanging two edges in $N_i$, as shown from (a$i$)
to (a$(i+1)$). (b$i$) computes $c(N_{i+1})-c(N_i)$, and illustrates
the relation of the terms used in the cost difference as a bipartite
graph, in which each edge represents their multiplication relation.}
\label{fgAlgDE4pattern}
\end{figure}

\begin{thm}
\label{thm-DE3-approx2}
\label{thm-DE4-approx2}%
There exist $O(\sqrt{n})$-approximation algorithms for DE3 and DE4,
which run in $O(n^2)$ time.
\end{thm}
\begin{proof}
Recall that given an unordered tree with fixed (resp., flexible) subwedges, the DE3 (DE4) problem is to find a circular permutation $\sigma$ of $\{1, \cdots, n\}$ (resp., a circular permutation $\sigma$ of $\{1, \cdots, n\}$ and a sub-wedge assignment $t$) so that the standard deviation of angles ($StdDev_{\sigma,t}$) is as small as possible.
We only concern DE4; the proof of DE3 is similar and simpler.

In what follows, we show that Algorithm~\ref{alg:DE4} correctly
produces $O(\sqrt{n})$-approximation solution in $O(n\log n)$ time.
Let $N_\opt$ be the matching for $V\times U$ witnessing the optimal
solution of the DE4 problem, and $N_\APX$ be the matching generated
by Algorithm~\ref{alg:DE4}. From Theorem~\ref{thm-DE4-approx1},
$c(N_\opt) \geq c_\LB$, and hence
\begin{eqnarray}
n \left( c(N_\opt) - \sum_{i=1}^n M_i m_i \right)%
&\geq& n \left( c_\LB - \sum_{i=1}^n M_i m_i \right) = n \sum_{e\in E_{T_\chi}} r_{\mu(e),\nu(e)} s_{\nu(e),\mu(e)} \nonumber \\
&\geq& n \sum_{i=1}^{\eta-1} \left( \sum_{j=1}^{|S_i'|-1} r_{j,j+1}' s_{j+1, j}' \right)\label{E-DE4-1}%
\end{eqnarray}
since $\mu(e) \neq \nu(e)$ for every edge $e \in E_\chi$. Observing
the matching $N_{S_i'}$ for each $S_i'$ generated by
Algorithm~\ref{alg:DE4} (e.g., see also Figure~\ref{fgAlgDE4}(a)),
without lose of generality, we assume that $1 \leq j_1 \leq j_2 \leq
\cdots \leq j_k \leq \cdots \leq j_h = |S_i'|-1$ and $h$ is odd such
that in $S_i'$, {\footnotesize
\begin{eqnarray}
\begin{array}{l@{\hspace{0.1mm}}l@{\hspace{0.1mm}}l@{\hspace{0.1mm}}l@{\hspace{0.1mm}}l@{\hspace{0.1mm}}l@{\hspace{0.1mm}}l@{\hspace{0.1mm}}}
                             & \ \ \ r_{1,2}'                 & \geq r_{2,3}' \geq \cdots & \geq r_{j_1,j_1+1}'         & \leq r_{j_1+1,j_1+2}'         & \leq \cdots \leq r_{j_2, j_2+1}'        & \geq \cdots \\
 \ \ \ \cdots &&& \ \ \ \cdots \\
\leq r_{j_k,j_k+1}'          & \geq r_{j_k+1,j_k+2}'          & \geq \cdots               & \geq r_{j_{k+1},j_{k+1}+1}' & \leq r_{j_{k+1}+1,j_{k+1}+2}' & \leq \cdots \leq r_{j_{k+2},j_{k+2}+1}' & \geq \cdots \\
 \ \ \ \cdots &&& \ \ \ \cdots\\
\leq r_{j_{h-1}, j_{h-1}+1}' & \geq r_{j_{h-1}+1, j_{h-1}+2}' & \geq \cdots               & \geq r_{j_h-1, j_h}'.%
\end{array}\label{E-jRelation}
\end{eqnarray}
}

Inequality (\ref{E-jRelation}) is explained as follows.
Since Line~17' in Algorithm~\ref{alg:DE4} considers the relationship between $r_{j,j+1}'$ and $r_{j+1,j+2}'$ for $j = 1, \cdots, |S_i'|-2$,
thus, without loss of generality, we use $h+1$ numbers (i.e., $1 \leq j_1 \leq j_2 \leq \cdots \leq j_k \leq \cdots \leq j_h = |S_i'|-1$) to classify all $r_{j,j+1}'$ data.
Then the data is alternately expressed as Inequality (\ref{E-jRelation}), in which $r_{j_1,j_1+1}'$, $r_{j_3,j_3+1}'$, \dots are local minimal; $r_{1,2}'$, $r_{j_2,j_2+1}'$, $r_{j_4,j_4+1}'$, \dots are local maximal.

Then, {\footnotesize
\begin{eqnarray}
c(N_{S_i'}) &=& (\sum_{j = 1}^{j_1-1} M_j' m_{j+1}' + M_{j_1+1}'m_1') + (M_{j_1}'m_{j_2+1}' + \sum_{j = j_1+1}^{j_2-1} M_{j+1}' m_{j}')\nonumber\\
                &&+ \cdots + (\sum_{j = j_k+1}^{j_{k+1}-1} M_j' m_{j+1}' + M_{j_{k+1}+1}'m_{j_k}') + (M_{j_{k+1}}'m_{j_{k+2}+1}' + \sum_{j = j_{k+1}+1}^{j_{k+2}-1} M_{j+1}' m_{j}')\nonumber\\
                &&+ \cdots + (\sum_{j = j_{h-1}+1}^{j_h-1} M_j' m_{j+1}' +
                M_{j_h}'m_{j_{h-1}}')\nonumber
\end{eqnarray}
}Therefore, 
{\footnotesize
\begin{eqnarray}
&&c(N_{S_i'}) - \sum_{j=1}^{|S_i'|} M_j' m_j'\nonumber\\
&&= (\sum_{j = 1}^{j_1} \sum_{l=j}^{j_1} r_{l,l+1}'s_{j+1,j}') + (\sum_{j = j_1+1}^{j_2} \sum_{l=j_1}^{j} r_{l,l+1}'s_{j+1,j}') - r_{j_2,j_2+1}'s_{j_2+1,j_2}'\nonumber\\
&&\ \ \ + \cdots + (\sum_{j = j_k}^{j_{k+1}} \sum_{l=j}^{j_{k+1}} r_{l,l+1}'s_{j+1,j}') + (\sum_{j = j_{k+1}+1}^{j_{k+2}} \sum_{l=j_{k+1}}^{j} r_{l,l+1}'s_{j+1,j}')- r_{j_{k+2},j_{k+2}+1}'s_{j_{k+2}+1,j_{k+2}}' \nonumber\\
&&\ \ \ + \cdots + (\sum_{j = j_{h-1}}^{j_h} \sum_{l=j}^{j_h} r_{l,l+1}'s_{j+1,j}') \nonumber%
\end{eqnarray}
}Consider Figure~\ref{fgAlgDE4}(b) for an example. The above
multiplication relationship of those $r_{\cdot,\cdot}$ and
$s_{\cdot,\cdot}$ for Figure~\ref{fgAlgDE4}(a) is given in
Figure~\ref{fgAlgDE4}(b). Figure~\ref{fgAlgDE4pattern} shows how to
transform from Figure~\ref{fgAlgDE4}(a) to Figure~\ref{fgAlgDE4}(b).

By Inequality~(\ref{E-jRelation}), since $r_{l,l+1}' \leq
r_{j+1,j}'$ for $j \leq l \leq j_1$ or $j_1 \leq l \leq j$ or
$\cdots$ or $j \leq l \leq j_{k+1}$ or $j_{k+1} \leq l \leq j$ or
$\cdots$ or $j \leq l \leq j_h$, we obtain:
\begin{eqnarray}
&&c(N_{S_i'}) - \sum_{j=1}^{|S_i'|} M_j' m_j' \nonumber\\
&&\leq (\sum_{j = 1}^{j_1} (j_1-j+1)  r_{j,j+1}' s_{j+1,j}') + (\sum_{j = j_1+1}^{j_2} (j-j_2+1)  r_{j,j+1}' s_{j+1,j}') - r_{j_2,j_2+1}'s_{j_2+1,j_2}'\nonumber\\
&&\ \ \ + \cdots + (\sum_{j = j_k}^{j_{k+1}} (j_{k+1}-j+1)  r_{j,j+1}' s_{j+1,j}') + (\sum_{j = j_{k+1}+1}^{j_{k+2}} (j-j_{k+2}+1)  r_{j,j+1}' s_{j+1,j}')\nonumber\\
&&\hspace{10cm} - r_{j_{k+2},j_{k+2}+1}'s_{j_{k+2}+1,j_{k+2}}'\nonumber\\
&&\ \ \ + \cdots + (\sum_{j = j_{h-1}}^{j_h} (j_{h}-j+1)  r_{j,j+1}' s_{j+1,j}')\nonumber\\
&&\leq n \left(\sum_{j = 1}^{|S_i'|-1}  r_{j,j+1}' s_{j+1,j}' \right) \label{E-DE4-2}%
\end{eqnarray}
Considering Figure~\ref{fgAlgDE4}(b) for an example, $c(N_{S_i'}) -
\sum_{j=1}^{|S_i'|} M_j' m_j' \leq 3 r_{1,2}' s_{2,1}' + 2 r_{2,3}'
s_{3,2}' + 1 r_{3,4}' s_{4,3}' + 2 r_{4,5}' s_{5,4}' + (3+5-1)
r_{5,6}' s_{6,5}' + 4 r_{6,7}' s_{7,6}' + 3 r_{7,8}' s_{8,7}' + 2
r_{8,9}' s_{9,8}' + 1 r_{9,10}' s_{10,9}' + 2 r_{10,11}' s_{11,10}'
+ 3 r_{11,12}' s_{12,11}' + (4 + 3 - 1) r_{12,13}' s_{13,12}' + 2
r_{13,14}' s_{14,13}' + 1 r_{14,15}' s_{15,14}'$.

By Inequalities~(\ref{E-DE4-1}) and (\ref{E-DE4-2}), we have
\begin{eqnarray}
n \left( c(N_\opt) - \sum_{i=1}^n M_i m_i \right)%
\geq c(N_\APX) - \sum_{i=1}^n M_i m_i \nonumber
\end{eqnarray}

In what follows, we explain how the algorithm runs in $O(n^2)$ time.
It suffices to explain Lines~13'--24'.
Lines~14' and 15' are just notations for the proof of correctness, not being executed.
In Line 17', $r_{j,j+1}'$ and $r_{j+1,j+2}'$ can be calculated in $O(1)$ time.
Hence, Lines~13'--24' in Algorithm~\ref{alg:DE4}
runs in $O(n)$ time, because the concerned availability (available maximum, minimum, second maximum, second minimum) is recorded and updated at each iteration in $O(1)$ time (noticing that $U$ and $V$ have been sorted, so has $S_i'$);
each element is recorded as the concerned availability at most $O(1)$ and matched only once.
\end{proof}

\section{Conclusion}
\label{sec:conclusion} This paper has investigated the tractability
of the problems for optimizing the angular resolution, the aspect
ratio, as well as the standard deviation of angles for balloon
drawings of ordered or unordered rooted trees with even sub-wedges
or uneven sub-wedges. It turns out that some of those problems are
NP-complete while the others can be solved in polynomial time. We
also give some approximation algorithms for those intractable
problems. A line of future work is to investigate the problems of
optimizing other aesthetic criteria of balloon drawings.

%
%
\bibliographystyle{elsart-num-sort}

\newpage

\section*{Appendix}




\noindent $\bullet$  {\bf On Proof of Theorem \ref{thm-DE4}}\\
Recall that the DE problem is concerned with minimizing the standard
deviation, which involves keeping all the angles as close to each
other as possible. Such an observation allows us to take advantage
of what is known for the 2SLW problem (which also involves finding a
circular permutation to bound a measure within  given lower and
upper bounds) to solve our problems. It turns out that, like 2SLW,
DE3 and and DE4 are NP-complete. Even though DE3, DE4 and 2SLW bear
a certain degree of similarity, a direct reduction from 2SLW to DE3
or DE4 does not seem obvious. Instead, we are able to tailor the
technique used for proving NP-hardness of 2SLW to showing DE3 and
DE4 to be NP-hard. To this end, we first briefly explain the
intuitive idea behind the NP-hardness proof of 2SLW shown in
\cite{V2003} to set the stage for our lower bound proofs.

The technique utilized in \cite{V2003} for the NP-hardness proof of
2SLW relies on  reducing from the {\em Hamiltonian-circle problem on
cubic graphs} (HC-CG)\footnote[1]{A cubic graph is a graph in which
every node  has degree three.}  (a known NP-complete problem). The
reduction is as follows. For a given  cubic graph $G$ with $n$
nodes, we construct a complete bipartite graph $\mathcal{B}(V,U)$
consisting of $n$ {\em blocks} in the following way. (For
convenience, $V$ (resp., $U$) is called the upper (resp., lower)
side.) For each node $v_i$ adjacent to $v_j$, $v_k$, $v_l$ in cubic
graph $G$, a {\em block} $\mathcal{B}_i$ of 14 nodes (7 on each
side)
is associated to $v_i$,   where the upper  side (resp., lower side)
contains three $v$-nodes (resp., $u$ nodes) corresponding to $v_j$,
$v_k$, $v_l$, and each side has a pair of $\lambda$-nodes, as well
as  a pair of $b$-nodes (as shown in Figure~\ref{fig:fg2SLW}). For
the three blocks $\mathcal{B}_j$, $\mathcal{B}_k$, and
$\mathcal{B}_l$ associated with nodes $v_j$, $v_k$, and $v_l$,
respectively, each has a $v$-node corresponding to $v_i$ (because
$v_i$ is adjacent to $v_j$, $v_k$, and $v_l$). These three $v$-nodes
are labelled as $v_{i1}$, $v_{i2}$, and $v_{i3}$. In the
construction, nodes in $V$ and $U$ correspond to those tasks to be
performed in stations $ST1$ and $ST2$, respectively, in 2SLW.

\begin{figure}[t]
\centering \scalebox{0.625}{\includegraphics*[0,0][416,183]{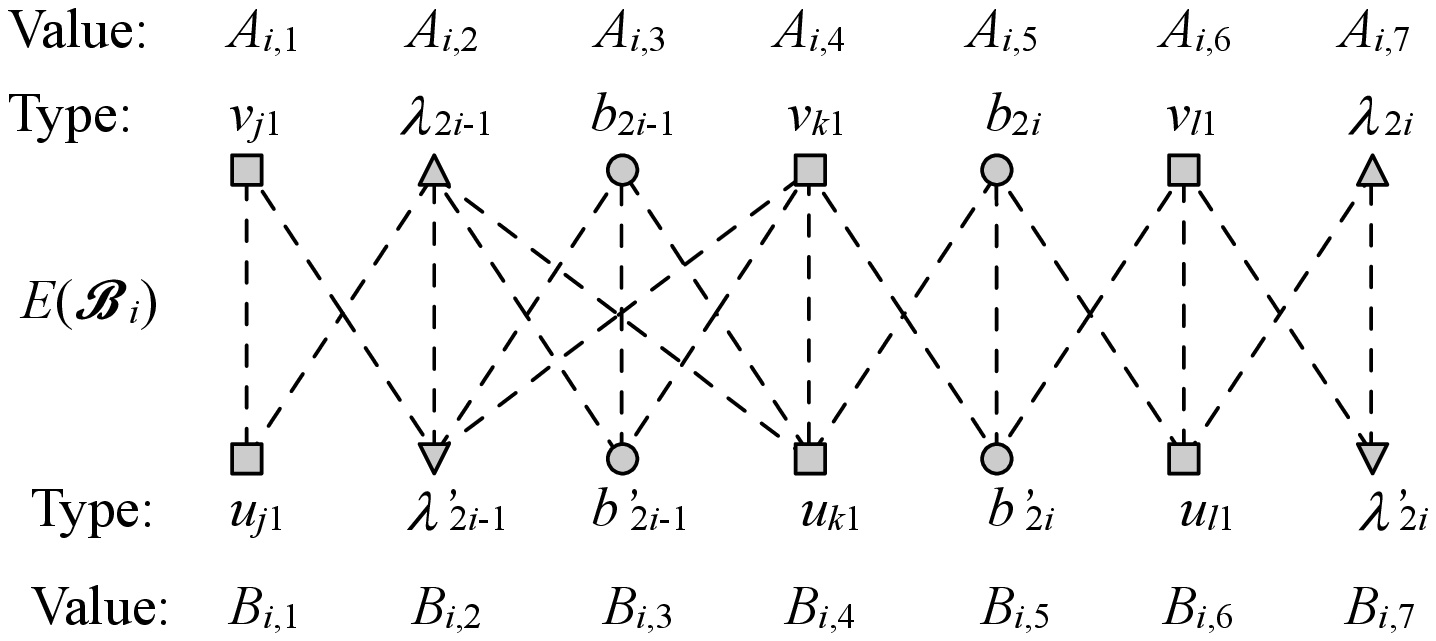}}%
\centering \caption{Illustration of reduction.} \label{fig:fg2SLW}
\bigskip
\bigskip
\centering \scalebox{0.625}{\includegraphics*[0,0][442,162]{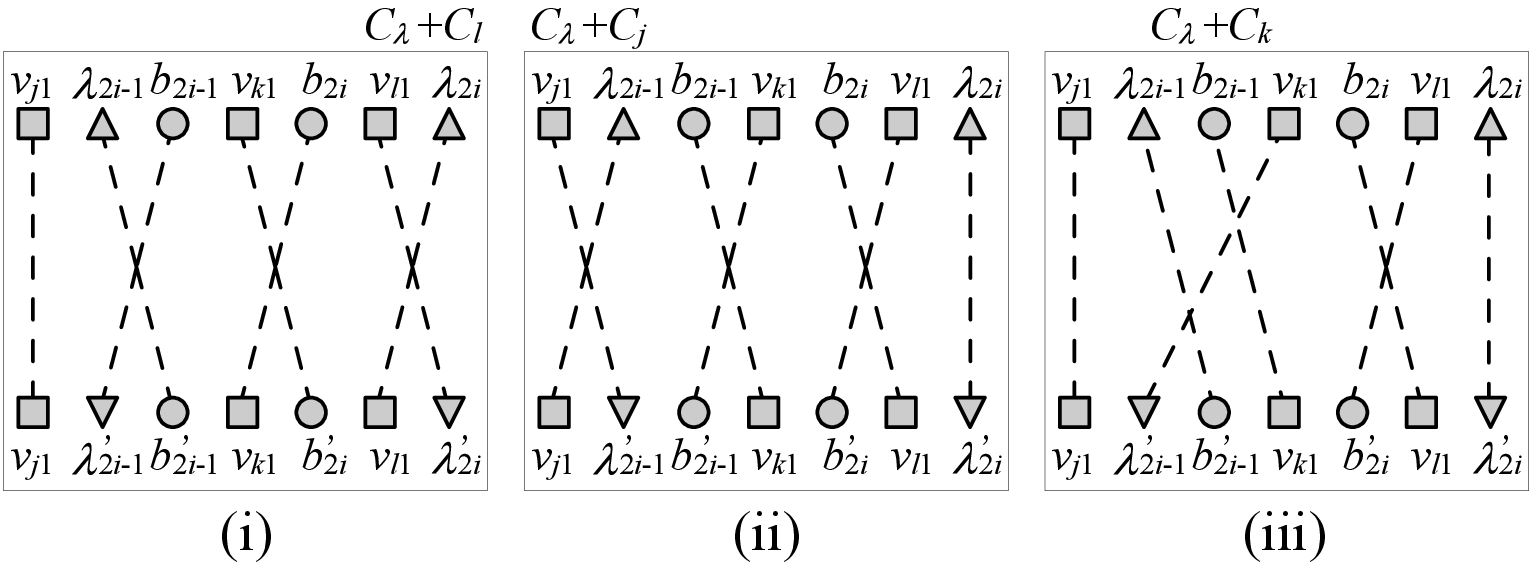}}%
\centering \caption{Three possible transition matchings for
$\mathcal{B}_i$.} \label{fig:fg4matching}
\end{figure}

As shown in Figure~\ref{fig:fg2SLW}, the nodes on the upper and
lower sides in $\mathcal{B}_i$ from the left to the right are
associated with the following values {\small
\begin{eqnarray}
&&(A_{i,1}, \cdots, A_{i,7}) =
(\kappa_i, \kappa_i-1, \kappa_i-2, \kappa_i-2, \kappa_i-3, \kappa_i-4, \kappa_i-5),\mbox{ and} \label{E-2SLW-valueA}\\
&&(B_{i,1}, \cdots, B_{i,7}) = (iK, iK+1, iK+2, iK+2, iK+3, iK+4,
iK+5),\label{E-2SLW-valueB}
\end{eqnarray}}respectively, where $\kappa_i = (n+1-i)K$ and $K$ is any integer
$\geq 7$; $LB=(n+1)K-1$ and $UB=(n+1)K+1$. Each edge in
$\mathcal{B}(V,U)$ has weight equal to the sum of the values of its
end points.

The instance of 2SLW consists of $7n$ jobs, in which $2n$ jobs
associated with pairs of $b$-nodes are $I_{01} = \{(b_{2i-1},
b_{2i-1}')$, $(b_{2i}, b_{2i}'): 1 \leq i \leq n\}$, $3n$ jobs
associated with $v$-nodes are $I_{02} = \{(v_{i1},u_{i2})$,
$(v_{i2}, u_{i3})$, $(v_{i3}, u_{i1}): i=1,...,n \}$, and $2n$ jobs
associated with pairs of $\lambda$-nodes are $I_{03} = \{(\lambda_i,
\lambda_{i\oplus 1}')$, $(\lambda_{i\oplus 1},\lambda_{i\oplus 2}'):
1 \leq i \leq n \}$. Note that $I_0 = I_{01} \cup I_{02} \cup
I_{03}$ is a perfect matching for $\mathcal{B}(V,U)$, and such a
matching is called a {\em city matching}.

The crux of the remaining construction is  based on the idea of
relating a permutation of the $7n$ jobs $(J_{[1]}, J_{[2]}, ...,
J_{[7n]})$ in the constructed 2SLW instance to a perfect matching
 in $\mathcal{B}(V,U)$ in such a
way that  $(W_{[i]2}, W_{[(i \ mod \ 7n)+1]1})$, $1 \leq i \leq 7n$,
are matches. Note that $W_{[i]2}, W_{[i+1]1}$ are the two tasks
performed by stations $ST1$ and $ST2$, respectively, simultaneously
at a certain time. One can easily observe that, because of bounds
$LB$ and $UB$,
 any matching $N$ as a solution for 2SLW cannot involve a edge connecting  two different
 blocks, and the only
edges which can be included in $N$ in each block are the dash lines
in Figure~\ref{fig:fg2SLW}.
%
Such a  perfect matching $N$ is called a {\em transition matching}.
If $I_0 \cup N$ forms a Hamiltonian cycle for $\mathcal{B}(V,U)$,
then it is called {\em complementary Hamiltonian cycle} (CHC).

We use notation $(\cdot,\cdot)$ (resp., $[\cdot,\cdot]$) to indicate
an edge of a city matching (resp., transition matching). Consider a
special transition matching $N_D = \{[A_{i,j},B_{i,j}]: i = 1, ...,
n, j = 1, ..., 7\}$. $I_0 \cup N_D$ consists of a master
$\lambda$-subcycle $\mathcal{C}_\lambda = [\lambda_{1},\lambda_1']
(\lambda_1',\lambda_{2}) ... [\lambda_{i},\lambda_i']
(\lambda_i',\lambda_{i\oplus 1}) ... [\lambda_{n},\lambda_n']
(\lambda_n',\lambda_1)$, $n$ $v$-subcycles $\mathcal{C}_i$ for
$i=1,...,n$ (e.g., $\mathcal{C}_1 =
[v_{11},u_{11}](u_{11},v_{12})[v_{12},u_{12}](u_{12},v_{13})[v_{13},u_{13}](u_{13},v_{11})$),
and $2n$ $b$-subcycles $\mathcal{C}_b = [b_i,b_i'](b_i,b_i')$ for
$i=1,...,n$. Hence, a CHC for $\mathcal{B}(V,U)$ is formed by
combining the $3n+1$ subcycles. From \cite{V2003}, in order to yield
a CHC for $\mathcal{B}(V,U)$, there are exactly three possible
transaction matchings for $\mathcal{B}_i$ as shown in
Figure~\ref{fig:fg4matching}. The design is such that edge $(v_i,
v_l)$ (resp., $(v_i, v_j)$ and $(v_i, v_k)$)  is in a HC of G if
Figure~\ref{fig:fg4matching}(i) (resp., (ii) and (iii)) is the
chosen permutation for the constructed 2SLW instance.   Following a
somewhat complicated argument, \cite{V2003} proved that there exists
a Hamiltonian cycle (HC) for the cubic graph $G$ if and only if
there exists a CHC for $\mathcal{B}(V,U)$, and such a CHC for
$\mathcal{B}(V,U)$ in turn suggest a sufficient and necessary
condition for  a solution for 2SLW.

\bigskip \noindent\underline{Proofs of Theorem~\ref{thm-DE4}}.
(Sketch) Now we are ready to show the theorem. We only consider the
DE4 problem; the DE3 problem is similar and in fact  simpler. Recall
that the DE4 problem is equivalent to finding a balloon drawing
optimizing $optSOP$. Consider the following  decision problem:
\begin{description}
    \item[$\mbox{\sc The DE4 Decision Problem}:$] %
Given  a star graph with flexible uneven angles specified by
Equation~(\ref{E-subWedge-uneven}) and an integer $UB$, determine
whether a drawing (i.e., specified by the permutation $\sigma \in
\Sigma$ and the assignments (0 or 1) for $t_i$ ($1\leq i \leq n$))
exists so that $SOP_{\sigma,t} \leq UB$.
\end{description}

It is obvious that  the problem is in NP; it remains to show
NP-hardness, which is established by a reduction from  HC-CG. In
spite of the similarity between our reduction  and the reduction
from HC-CG to  2SLW (\cite{V2003})  explained earlier,  the
correctness proof of our reduction is a lot more complicated than
the latter, as we shall explain in detail shortly.

In the new setting, Equations~(\ref{E-2SLW-valueA}) and
(\ref{E-2SLW-valueB}) become: {\small
\begin{eqnarray}
&&(A_{i,1}, \cdots, A_{i,7}) =
(\kappa(i), \kappa(i)-2, \kappa(i)-3, \kappa(i)-4, \kappa(i)-6, \kappa(i)-8, \kappa(i)-9);\nonumber\\
&&(B_{i,1}, \cdots, B_{i,7}) = (9ni, 9ni+1, 9ni+2, 9ni+3, 9ni+5, 9ni+7, 9ni+9) \nonumber
\end{eqnarray}}for $i = 1, 2, ..., n$ where $\kappa(i) = 9n(2n+2-i)$, $n \geq 2$,
and $UB=\sum_{i=1}^{7n}M_i m_i+7n$ where $M_i$ (resp., $m_i$) is the
$i$-th maximum (resp., minimum) among the $14n$ values.  (Note that
such a setting satisfies the premise of
Inequality~(\ref{E-lemma-perm-2}) in Lemma~\ref{lemma-perm}, and
hence can utilize the inequality.) Hence, we have that:
\begin{eqnarray}
A_{i,j} &>& B_{i,j}, \mbox{ for any $i,j$;}\nonumber\\
A_{i,j} &>& A_{k,l} \mbox{ and } B_{i,j} < B_{k,l} \mbox{ if ($i <
k$) or ($i = k$ and $j < l$).}\nonumber
\end{eqnarray}
Note that the above implies that the $j$-th upper (resp., lower)
node in $\mathcal{B}_i$ is $M_{7i+j}$ (resp., $m_{7i+j}$) for $i\in
\{1,...,n\}$ and $j \in \{1,...,7\}$. Define $r_{a,b} = M_{7i+a} -
M_{7i+b}$ and $s_{a,b} = m_{7i+a} - m_{7i+b}$ in $\mathcal{B}_i$.
Hence,
\[
r_{1,2}s_{2,1} = 2, \ r_{2,3}s_{3,2} = 1, \ r_{3,4}s_{4,3} = 1,
r_{4,5}s_{5,4} = 4, \ r_{5,6}s_{6,5} = 4, \ r_{6,7}s_{7,6} = 2,
\]
which are often utilized throughout the remaining proof.

If $\Omega$ is a set of transition edges, the sum of the transition
edge weights is denoted by $c(\Omega)$. If $C_H = I_0 \cup N$ is a
CHC for $\mathcal{B}(V,U)$ where $I_0$ (resp., $N$) is the city
matching (resp., transition matching) of the CHC and $t_i = 0$ for
$i=1,...,n$ (i.e., flipping sub-wedges is not allowed), then $c(N) =
\sum_{e \in N} c(e) = SOP_{\sigma,t}$ where $c(e)$ is the weight of
the transition edge $e$.

Now based on the above setting, we show that there exists a HC for
the cubic graph $G$ if and only if there exists a CHC $C_H = I_0
\cup N$ for the instance $\mathcal{B}(V,U)$ of the DE4 problem such
that $c(N) \leq UB$.

Suppose that $G$ has a Hamiltonian cycle $C_H$. Let $ C_H =
v_{[1]},v_{[2]}, ..., v_{[n]}$. The construction of a solution for
$\mathcal{B}(U,V)$ is the same as \cite{V2003},  as explained in the
following. Initiating with $\mathcal{B}_{[1]}$, there exists a pair
($u_{[2]}, v_{[2]}$) of nodes in $V \times U$ corresponding to
$v_{[2]} \in G$ because $v_{[1]}$ is connected with $v_{[2]}$. From
\cite{V2003}, we have that $\mathcal{C}_\lambda$ is merged with
$\mathcal{C}_l$, $\mathcal{C}_j$, and $\mathcal{C}_k$ respectively
in Figure~\ref{fig:fg4matching} (i), (ii), and (iii). Hence,
considering the order of $B_{[1]},B_{[2]},...,B_{[n]}$, in iteration
$i$, by choosing the appropriate transition matching, say $N_{[i]}$,
of $\mathcal{B}_i$ from the three possible matchings in
Figure~\ref{fig:fg4matching}, $N_{[i]}$ merges
$\mathcal{C}_{[i\oplus 1]}$ with the master subcycle
$\mathcal{C}_{\lambda}$. Besides, since the two $b$-subcycles in
each $\mathcal{B}_i$ also are merged with $\mathcal{C}_\lambda$ in
any matching of Figure~\ref{fig:fg4matching}, we can obtain a
complementary cycle traversing all nodes in $\mathcal{B}(U,V)$.

We need to check $c(N) \leq UB$. In fact, we show that $c(N) = UB$
as follows. It suffices to show that $c(N_i) = \sum_{j=1}^7
M_{7i+j}m_{7i+j} + 7$ for any $i \in \{1,...,n\}$ where $N_i$ is the
transition matching for $\mathcal{B}_i$. Denote $\Delta c(N_i) =
c(N_i) - \sum_{j=1}^7 M_{7i+j}m_{7i+j}$. We can prove that $\Delta
c(N_i) = 7$ for every matching in Figure~\ref{fig:fg4matching}.
Case~(i) is shown as follows, and the others are similar: {\small
\begin{eqnarray}
\Delta c(N_i) &=& M_{7i+1}m_{7i+1} + M_{7i+2}m_{7i+3} +
M_{7i+3}m_{7i+2} + M_{7i+4}m_{7i+5} \nonumber\\
&&  + M_{7i+5}m_{7i+4} + M_{7i+6}m_{7i+7} + M_{7i+7}m_{7i+6} -
\sum_{j=1}^7
M_{7i+j}m_{7i+j} \label{E-DE4proof-if}\\
&=& r_{2,3}s_{3,2} +
r_{4,5}s_{5,4} + r_{6,7}s_{7,6} 
= 1 + 4 + 2 = 7\nonumber
\end{eqnarray}
}From the above computation, one should notice that if $M_j$ is
matched with a sub-wedge larger than $m_j$ and $M_{j+1}$ is matched
with a sub-wedge less than $m_{j+1}$ for $j \in
\{7i+1,7i+2,...,7i+6\}$, then $\Delta c(N_i)$ includes
$r_{j,j+1}s_{j+1,j}$.

%



The converse, i.e., showing the existence of a CHC $C_H = I_0 \cup
N$ for the instance $\mathcal{B}(V,U)$ of   DE4  with $c(N) \leq UB$
implies the presence of a HC in $G$, is rather complicated. The key
relies on the following three claims.
\begin{enumerate}
    \item[(S-1)] (\emph{Bipartite}) There are no transition edges in $N$ between any pairs of upper (resp., lower) nodes in $C_H$.
    \item[(S-2)] (\emph{Block}) There are no transition edges in $N$ between two blocks in $C_H$.
    \item[(S-3)] (\emph{Matching}) There is only one of $\mathcal{C}_j$, $\mathcal{C}_k$, and $\mathcal{C}_l$ merged with the master subcycle $\mathcal{C}_\lambda$ in
    each $\mathcal{B}_i$. (Recall that each node $v_i$ is adjacent to $v_j$, $v_k$, $v_l$ in $G$, and hence the statement implies the presence of a HC in $G$.)
\end{enumerate}


For proving the above statements, we need the following claims:
\begin{description}
    \item[Claim 1] (see \cite{KS2009}) Given two transition matchings $N$
and $N'$ between $V$ and $U$, there exists a sequence of exchanges
which transforms $N$ to $N'$.
    \item[Claim 2] If $N$ is a transition matching between $V$ and $U$ and involves two
    edges $e_1$ and $e_2$ crossing each other, then $c(N) > c(N')$ for $N'=N\otimes
    (e_1,e_2)$.
\end{description}
(Claim 2 can be proved by easily checking $c(N)-c(N')>0$.) It is
very important to notice that Claim 2 can be adapted even when $I_0
\cup N$ may NOT be a CHC. The transition matching where $M_j$ is
matched with $m_j$ for every $j$ (every transition edge is visually
vertical) is denoted by $N_D$, i.e., $c(N_D) = \sum_{j=1}^{7n} M_j
m_j$. Note that if each edge in $N$ is between $V$ and $U$, we can
obtain $c(N) > c(N_D)$ by repeatedly using Claim 2 in the order from
the leftmost node to the rightmost node of $V$, similar to the
technique in the proof of Claim 1 \cite{KS2009}. \qed

\vspace{0.1cm}\underline{Proof of Statement (S-1)}. Supposing that
there exits $\overline{k} \geq 1$ transition edges between pairs of
upper nodes in $C_H$, then there must exist $\overline{k}$
transition edges between pairs of lower nodes in $C_H$, by
Pigeonhole Principle. Select one of the upper (resp., lower)
transition edges, say $e_1 = (M_a, M_b)$, (resp., say $e_2 = (m_p,
m_q)$). Consider $N' = N \otimes (e_1, e_2)$. Then $c(N) -
c(N')=(M_a-m_p)(M_b-m_q) \geq (M_n - m_n)^2 = 18^2(n-1)^2
> 7n$. Hence, $c(N) > c(N') + 7n$. By the same
technique, we can find $N''$ where each edge in $N''$ is between $U$
and $V$ such that $c(N) > c(N'')+7\overline{k}n \geq c(N'')+7n \geq
c(N_D) + 7n = UB$, which is impossible.\qed

\vspace{0.1cm}\underline{Proof of Statement (S-2)}. By Statement
(S-1), each edge in the transition matching of $C_H$ is between $V$
and $U$. Suppose there exists at least one transition edge between
two blocks. Assume there are $l$ blocks, $\{\mathcal{B}_{k_1},
\mathcal{B}_{k_2}, ..., \mathcal{B}_{k_l}\}$, with transition edges
across two blocks. Let $k_{min} = \min(k_1, k_2, ..., k_l)$.
Consider $e_1 = (M_a, m_d)$ is the transition edge between
$\mathcal{B}_{k_{min}}$ and $\mathcal{B}_{k_{i}}$ for $i \in \{1,
...,l\}$, and $k_{min} \neq k_i$. Then there must exist a transition
edge connecting to one of the lower nodes of
$\mathcal{B}_{k_{min}}$, say $m_c$, by Pigeonhole Principle, and we
say the edge $e_2 = (M_b, m_c)$ where $m_c$ and $M_b$ are
respectively from $\mathcal{B}_{k_{min}}$ and $\mathcal{B}_{k_j}$
for $j \in \{1, ...,l\}$ and $k_j \neq k_{min}$. Note that $e_2$
must cross $e_1$ because $M_a$ and $m_c$ are in
$\mathcal{B}_{k_{min}}$, i.e., $M_a
> M_b$ and $m_c < m_d$. Besides, we have $M_a
\geq M_b + 9n - 9$ and $m_d \geq m_c + 9n - 9$ because two end
points of edge belong to different blocks. Consider $N'=N\otimes
(e_1, e_2)$. Then $c(N) - c(N') = (M_a-M_b)(m_d-m_c) \geq (9n-9)^2
> 7n$ for $n \geq 2$. That is, $c(N) > c(N') + 7n \geq c(N_D) +
7n$, which is a contradiction.\qed

\vspace{0.1cm}\underline{Proof of Statement (S-3)}. Recall that $I_0
\cup N_D$ in every $\mathcal{B}_i$ involves subcycles
$\mathcal{C}_j$, $\mathcal{C}_\lambda$, $\mathcal{C}_{b_{2i-1}}$,
$\mathcal{C}_k$, $\mathcal{C}_{b_{2i}}$, $\mathcal{C}_l$,
$\mathcal{C}_\lambda$ from the leftmost to the rightmost. If there
exists a CHC $C_H = I_0 \cup N$ for the instance $\mathcal{B}(V,U)$,
each $b$-subcycle in $\mathcal{B}_i$ has to be merged with some
subcycle in the same $\mathcal{B}_i$ by Statements (S-1) and (S-2).
$\Delta c(N_i)$ is at least 5 due to the merging of $b$-subcycles
from the following four cases (here it suffice to discuss the
merging of $b$-subcycles with their adjacent subcycles because
$\Delta c(N_i)$ in others cases are larger):
\begin{enumerate}
    \item $\mathcal{C}_{b_{2i-1}}$ merged with $\mathcal{C}_\lambda$ and $\mathcal{C}_{b_{2i}}$ merged with $\mathcal{C}_k$: 
$\Delta c(N_i) > r_{2,3}s_{3,2}+r_{4,5}s_{5,4} = 1 + 4 =5$%
    \item $\mathcal{C}_{b_{2i-1}}$ merged with $\mathcal{C}_\lambda$ and $\mathcal{C}_{b_{2i}}$ merged with $\mathcal{C}_l$: 
$\Delta c(N_i) > r_{2,3}s_{3,2}+r_{5,6}s_{6,5} = 1 + 4 =5$%
    \item $\mathcal{C}_{b_{2i-1}}$ merged with $\mathcal{C}_k$ and $\mathcal{C}_{b_{2i}}$ merged with $\mathcal{C}_k$: 
$\Delta c(N_i) > r_{3,4}s_{4,3}+r_{4,5}s_{5,4} = 1 + 4 =5$%
    \item $\mathcal{C}_{b_{2i-1}}$ merged with $\mathcal{C}_k$ and $\mathcal{C}_{b_{2i}}$ merged with $\mathcal{C}_l$: 
$\Delta c(N_i) > r_{3,4}s_{4,3}+r_{5,6}s_{6,5} = 1 + 4 =5$%
\end{enumerate}

Recall that there are $3n + 1$ subcycles in $\mathcal{B}$. Hence we
require at least $3n$ times of merging subcycles to ensure these
subcycles to be merged as a CHC. Since we have discussed that two
$b$-subcycles have to be merged in each $\mathcal{B}_i$ (i.e., the
total times of merging $b$-subcycles are $2n$), we require at least
$n$ more times of merging subcycles to obtain a CHC. In fact, the
$n$ times of merging subcycles is because each $\mathcal{B}_i$
contributes once of merging subcycles. As a result, Statement (S-3)
is proved if we can show that after merging two $b$-subcycles in
each $\mathcal{B}_i$, the third merging subcycles in $\mathcal{B}_i$
is to merge one of $C_j$, $C_k$, and $C_l$ with $C_\lambda$.

In what follows, we discuss $\Delta c(N_i)$ when there are exactly
$\overline{h}$ times of merging subcycles in $N_i$:
\begin{itemize}
    \item If $\overline{h} = 2$, then $\Delta c(N_i) > 5$.
    \item If $\overline{h} = 3$ and the transition matching of $\mathcal{B}_i$ is one of the matchings in Figure~\ref{fig:fg4matching}, then $\Delta c(N_i) = 7$.
    \item If $\overline{h} = 3$ and the transition matching of $\mathcal{B}_i$ is NOT any of the matchings in Figure~\ref{fig:fg4matching}, then $\Delta c(N_i) > 7$.
    \item If $\overline{h} = 4$, then $\Delta c(N_i) > 9$.
    \item If $\overline{h} = 5$, then $\Delta c(N_i) > 11$.
    \item If $\overline{h} = 6$, then $\Delta c(N_i) > 13$.
\end{itemize}
If the above statements on $\overline{h}$ hold, then Statement (S-3)
hold. The reason is as follows. Remind that we need $3n$ times of
merging subcycles to be a CHC. Therefore, if there exists a
transition matching of $\mathcal{B}_i$ with $\overline{h}=2$ for
some $i$ (i.e., there are exactly two times of merging subcycles in
$\mathcal{B}_i$), then there must exists a $\mathcal{B}_j$ for some
$j$ with $\overline{h} \geq 4$. Then $\Delta c(N_i)+\Delta c(N_j)
> 14$, which is impossible because this results in the total
$\Delta c$ larger than $7n$.\qed

\vspace{0.1cm}\underline{Proof of Statements on $\overline{h}$}.
Note that the transition matching of every $\mathcal{B}_i$ can be
viewed as a permutation of $\{ M_{7i+1}, M_{7i+2}, ..., M_{7i+7}\}$
(a mapping from $V$ to $U$), and hence different ordering or
different times of merging subcycles lead to a permutation with
different factors, e.g, the permutation for
Figure~\ref{fig:fg4matching}(i) is $\langle M_{7i+1}\rangle \langle
M_{7i+2} M_{7i+3}\rangle \langle M_{7i+4} M_{7i+5}\rangle \langle
M_{7i+6} M_{7i+7}\rangle$. If we let $f = \langle M_{j_1}, M_{j_2},
..., M_{j_h}\rangle$ be a nontrivial factor of the permutation for
$N_i$, then $c(N_i) \geq c(f) \geq \sum_{k=j_1}^{j_h} M_{k} m_{k} +
\sum_{k=j_1}^{j_{h-1}} r_{k,k+1}s_{k+1,k}$ by
Equation~(\ref{E-lemma-perm-1}) in Lemma~\ref{lemma-perm}. Here we
concern the value $\sum_{k=j_1}^{j_{h-1}} r_{k,k+1}s_{k+1,k}$
induced by $f$ (which is denoted by $\Theta(f)$) because it can be
viewed as a lower bound of $\Delta c(N_i)$.

If a factor $f$ includes $M_j$ but excludes $M_{j+1}$, then we say
that $f$ has a {\em lack} at $M_{j+1}$. We observe that if the
permutation $p_i$ for $\mathcal{B}_i$ has a lack, then we can find a
permutation $p_i'$ for $\mathcal{B}_i$ consisting of the factors
without any lacks such that $\Theta(p_i') < \Theta(p_i)$ in which
the number of factors of $p_i'$ is the same as that of $p_i$ and the
size of each factor is also the same. The reason is as follows.
Assume that $p_i$ has a factor $f = \langle ..., M_j, M_l,
...\rangle$ with a lack at $M_{j+1}$ (i.e., $l\neq j+1$) and the
minimum number appearing in the factor is $M_q$. Let $p_i'$ be
almost the same as $p_i$ except the factor $f$ in $p_i$ is modified
as a factor without any lacks involving $M_{j+1}$ but excluding
$M_q$ in $p_i'$. Then by Equation~(\ref{E-lemma-perm-1}) in
Lemma~\ref{lemma-perm}, $\Theta(p_i)-\Theta(p_i') \geq
(r_{j,j+2}s_{j+2,j} - r_{j,j+1}s_{j+1,j} - r_{j+1,j+2}s_{j+2,j+1}) +
r_{q-1,q}s_{q,q-1} \geq 2 + 1 > 0$. In the similar way, we can find
a permutation with factors without any lacks.

In light of the above, it suffices to consider the permutation for
$\mathcal{B}_i$ consisting of the factors without any lacks when
discussing the lower bound of $c(N_i)$. Thus, in the following, when
we say that the permutation for $\mathcal{B}_i$ has a factor $f$,
this implies that $f$ has no lacks, so Lemma~\ref{lemma-perm} can be
applied to $f$.

Now we are ready to prove the statements on $\overline{h}$. The
statement of $\overline{h}=2$ holds because $c(N_i)$ is increased by
at least 5 when two $b$-subcycles have to be merged in each
$\mathcal{B}_i$. As for the statement of $\overline{h}=6$, note that
merging six subcycles implies a permutation with a factor of size
seven. Thus, by Equation~(\ref{E-lemma-perm-1}) in
Lemma~\ref{lemma-perm}, $\Delta c(N_i) \geq \sum_{j=1}^6 r_{j,j+1}
s_{j+1,j} = 2 + 1 + 1 + 4 + 4 + 2 = 14
> 13$, as required. Let $\psi = \sum_{j=1}^6 r_{j,j+1} s_{j+1,j} = 14$
for the convenience of the following discussion. As for the
statement of $\overline{h}=5$, the permutation involves two
nontrivial factors after five times of merging subcycles. Note that
one of the two factors has size at least four, and hence the factor
$\langle j_1, ..., j_4\rangle$ contributes $\sum_{k=j_1}^{j_4} M_k
m_k + \sum_{k=j_1}^{j_4} r_{k,k+1} s_{k+1,k} + (4-2)$ by
Equation~(\ref{E-lemma-perm-2}) in Lemma~\ref{lemma-perm}.
Therefore, by Equation~(\ref{E-lemma-perm-1}) in
Lemma~\ref{lemma-perm}, $\Delta c(N_i) \geq \psi -
r_{x,x+1}s_{x+1,x} + (4-2) = 16 - r_{x,x+1}s_{x+1,x}$ for some $x
\in \{1,...,6\}$. (Note that $-r_{x,x+1}s_{x+1,x}$ suggests that
$M_x$ and $M_{x+1}$ are in different factors.) Since
$r_{x,x+1}s_{x+1,x} \leq 4$, hence $\Delta c(N_i) \geq 12 > 11$, as
required.

As for the statement of $\overline{h}=4$, by
Equation~(\ref{E-lemma-perm-1}), $\Delta c(N_i) \geq \psi -
r_{x,x+1} s_{x+1,x} - r_{y,y+1} s_{y+1,y}$ for some $x,y \in
\{1,...,6\}$ and $x \neq y$. Discuss all possible cases of pair
$(x,y)$ as follows. Consider one of $x,y$ is 2 or 3. We assume that
$x = 2$, and the other case is similar. Hence,
$r_{x,x+1}s_{x+1,x}=1$. Since $r_{y,y+1}s_{y+1,y} \leq 4$ and there
exists a factor with size at least three in this case, $\Delta
c(N_i) \geq \psi - 1 - 4 + (3-2) = 10 > 9$ by
Equation~(\ref{E-lemma-perm-2}), as required. The remaining cases
are $(1,4), (1,5), (1,6), (4,5), (4,6)$, and $(5,6)$. Consider one
of $x,y$ is 1 or 6. We assume that $x = 1$, and the other case is
similar. Hence $r_{x,x+1}s_{x+1,x}=2$. Since $r_{y,y+1}s_{y+1,y}
\leq 4$ and there exists a factor with size at least four or two
factors with size at least three in this case, $\Delta c(N_i) \geq
\psi - 2 - 4 + (4-2) \vee 2(3-2) = 10
> 9$ by Equation~(\ref{E-lemma-perm-2}), as required. Last, consider $(x,y) =
(4,5)$, namely, $M_4$ and $M_5$ (resp., $M_5$ and $M_6$) are in
different factors. Hence, $M_5$ cannot be matched with $m_4$ nor
$m_6$, i.e., subcycle $\mathcal{C}_{b_{2i}}$ cannot be merged with
adjacent subcycles $\mathcal{C}_k$, $\mathcal{C}_l$. Since merging
$\mathcal{C}_{b_{2i}}$ with $\mathcal{C}_{b_{2i-1}}$ induces the
smallest cost $r_{3,5}s_{53}=9$ in this case, and the other two
times of merging subcycles must induce cost more than 2, hence
$\Delta c(N_i)$ is at least 9.

As for the two statements of $\overline{h}=3$, by
Equation~(\ref{E-DE4proof-if}), $\Delta c(N_i)$ in the case when
$N_i$ is one of the matchings in Figure~\ref{fig:fg4matching} is
exactly seven, as required. Then we consider the case when $N_i$ is
not in Figure~\ref{fig:fg4matching} in the following. By
Equation~(\ref{E-lemma-perm-1}), $\Delta c(N_i) \geq
r_{x,x+1}s_{x+1,x} + r_{y,y+1}s_{y+1,y} +r_{z,z+1}s_{z+1,z} \geq 5 +
r_{z,z+1}s_{z+1,z}$ for some $x,y,z\in\{1,...,6\}$ and $x\neq y \neq
z \neq x$ since it is necessary to merge $b$-subcycles, which
contributes at least 5. It suffices to consider the cases when
$r_{z,z+1}s_{z+1,z} \leq 2$, which may violate our required. That
is, $z$ may be 1, 2, 3 or 6. By considering four possible cases of
merging $b$-subcycles, one may easily check that whatever $z$ is,
$\Delta c(N_i)$ must be either larger than 7 or in
Figure~\ref{fig:fg4matching}.\qed

\end{document}